\documentclass[12pt,table]{article}

 \usepackage{frankstyle}

\newcommand{\E}{\mathbb{E}}
\newcommand{\I}{\mathbbm{I}}

\newcommand{\var}{\text{Var}}
\newcommand{\cov}{\text{Cov}}
\newcommand{\tvec}{\text{vec}}
\newcommand{\tr}{\text{tr}}

\newcommand{\normal}{\text{Normal}}
\newcommand{\IW}{\text{Inverse-Wishart}}
\newcommand{\IG}{\text{Inverse-Gamma}}
\newcommand{\SB}{\text{SB}}
\newcommand{\SNR}{\text{SNR}}
\newcommand{\naive}{{\text{naive}}}
\newcommand{\BDML}{{\text{BDML}}}
\newcommand{\FDML}{{\text{FDML}}}
\newcommand{\pto}{\overset{p}{\to}}
\newcommand{\dto}{\overset{d}{\to}}
\newcommand{\pstarto}{\overset{P^*}{\to}}
\newcommand{\sumi}{\sum_{i=1}^n}
\newcommand{\avgi}{\frac 1 n\sum_{i=1}^n}
\newcommand{\sqrti}{\frac 1 {\sqrt n}\sum_{i=1}^n}

\title{Bayesian Double Machine Learning for Causal Inference\thanks{We thank Otilia Boldea, Neal Shephard, and seminar participants at York and the Frank Diebold 65+35 Conference for helpful comments and suggestions. We thank Lasse von der Heydt for excellent research assistance.}}

 \author[1]{Francis J.\ DiTraglia\thanks{DiTraglia: \href{mailto:francis.ditraglia@economics.ox.ac.uk}{francis.ditraglia@economics.ox.ac.uk}}}
 \author[2]{Laura Liu\thanks{Liu: \href{mailto:laura.liu@pitt.edu}{laura.liu@pitt.edu}}}
 \affil[1]{\normalsize Department of Economics, University of Oxford}
 \affil[2]{\normalsize Department of Economics, University of Pittsburgh}

\date{\small First Version: November 6, 2024 \quad This Version: \today} 

\begin{document}

% Create un-numbered title page
\clearpage
\maketitle
\thispagestyle{empty}

%\vspace{-0.8cm}

\begin{abstract}
  \singlespacing
	%!TEX root = ./main.tex

This paper proposes a simple, novel, and fully-Bayesian approach for causal inference in partially linear models with high-dimensional control variables. Off-the-shelf machine learning methods can introduce biases in the causal parameter known as regularization-induced confounding. To address this, we propose a Bayesian Double Machine Learning (BDML) method, which modifies a standard Bayesian multivariate regression model and recovers the causal effect of interest from the reduced-form covariance matrix. Our BDML is related to the burgeoning frequentist literature on DML while addressing its limitations in finite-sample inference. Moreover, the BDML is based on a \textit{fully generative probability model} in the DML context, adhering to the likelihood principle. We show that in high dimensional setups the na\"ve estimator implicitly assumes no selection on observables---unlike our BDML. The BDML exhibits lower asymptotic bias and achieves asymptotic normality and semiparametric efficiency as established by a Bernstein–von Mises theorem, thereby ensuring robustness to misspecification. In simulations, our BDML achieves lower RMSE, better frequentist coverage, and shorter confidence interval width than alternatives from the literature, both Bayesian and frequentist.

% For pasting, removed latex syntax
% This paper proposes a simple, novel, and fully-Bayesian approach for causal inference in partially linear models with high-dimensional control variables. Off-the-shelf machine learning methods can introduce biases in the causal parameter known as regularization-induced confounding. To address this, we propose a Bayesian Double Machine Learning (BDML) method, which modifies a standard Bayesian multivariate regression model and recovers the causal effect of interest from the reduced-form covariance matrix. Our BDML is related to the burgeoning frequentist literature on DML while addressing its limitations in finite-sample inference. Moreover, the BDML is based on a fully generative probability model in the DML context, adhering to the likelihood principle. We show that in high dimensional setups the naive estimator implicitly assumes no selection on observables--unlike our BDML. The BDML exhibits lower asymptotic bias and achieves asymptotic normality and semiparametric efficiency as established by a Bernstein–von Mises theorem, thereby ensuring robustness to misspecification. In simulations, our BDML achieves lower RMSE, better frequentist coverage, and shorter confidence interval width than alternatives from the literature, both Bayesian and frequentist.


\end{abstract}

\onehalfspacing

\newpage

%!TEX root = ./main.tex

\section{Introduction}

Two decades into the ``credibility revolution'' popularized by \emph{Mostly Harmless Econometrics}, much work in applied microeconomics continues to rely on relatively simple reduced-form models to untangle cause-and-effect from observational data.
In the simplest case, and the one we consider below, a researcher argues that a treatment $D$ is as good as randomly assigned after adjusting for a set of control variables $X$ and goes on to report the coefficient on $D$ from a regression of $Y$ on $(D,X)$ as the causal effect of interest.
Naturally this selection-on-observables approach is only as good as the choice of controls; researchers who rely on it typically devote much of their paper to defending the choice of $X$ in their application.

All else equal, a selection-on-observables strategy is more convincing when it includes more observables.
For this reason the dimension of $X$ can be very large relative to sample size.
Even if there are relatively few ``underlying'' control variables, fully adjusting for them may require us to include a large number of interactions and other transformations in our regression.
Applied researchers thus face a difficult trade-off: including more controls reduces the risk of bias but can lead to extremely imprecise estimates of the causal effect of interest.

In this paper, we provide a simple, fully Bayesian approach to navigating this bias-variance trade-off in a workhorse partially linear model: Bayesian Double Machine Learning (BDML).
While our methods are Bayesian, the resulting point and interval estimators are designed to appeal to a broad audience. 
For Bayesians, BDML offers a natural way to bring the full Bayesian toolbox to bear on high-dimensional causal inference problems in applied microeconomics.
For Frequentists, BDML provides superior finite-sample performance in practice--lower root mean-squared error and improved coverage of confidence intervals--while delivering the same attractive asymptotic guarantees as existing Frequentist approaches. 

Any attempt to tame the estimation uncertainty that arises from including many controls relies on \emph{regularization}, also known as shrinkage.
Rather than ``projecting out'' $X$ from $Y$ and $D$ using ordinary least squares (OLS), a regularized estimator \emph{constrains} the high-dimensional vector $\beta$ of coefficients on $X$ in the regression of $Y$ on $(D,X)$. 
Bayesian approaches do so by placing a prior on $\beta$; machine learning (ML) approaches do so either by adding a penalty term to the objective function to favor ``smaller'' values for the elements of $\beta$, as in LASSO or Ridge, or by penalizing model complexity through implicit algorithmic constraints, as in random forests and gradient boosting.

Whether Bayesian or Frequentist, implicit or explicit, any application of regularization to causal inference problems faces a key challenge: \emph{regularization-induced confounding} (RIC).\footnote{This term was coined by \cite{hahn_regularization_2018}.}
Unlike OLS, regularized estimators only partially adjust $D$ and $Y$ for the observed confounders $X$.
This induces a correlation between $D$ and the regularized regression residuals, biasing the estimated causal effect.
Granted, researchers must be prepared to accept at least some increase in bias in exchange for reduced estimation uncertainty.
The problem of RIC is that this bias can be \emph{extremely sensitive} to the precise way in which regularization is carried out, depending on the relationship between $D$ and $X$.
Indeed, na\"{i}ve approaches to regularized estimation can easily \emph{increase} the root mean-squared error (RMSE) of our causal effect estimator, adding a bias that far outweighs the reduction in variance.
Crucially, our BDML approach side-steps this problem completely: using a simple re-parameterization of the model, we avoid RIC while relying on standard methods for Bayesian linear regression.

From a subjective Bayesian point of view, RIC may seem puzzling. 
If we believe our prior and likelihood, Bayes optimality \emph{requires} us to base estimation and inference on the implied posterior.
In practice, however, priors are often chosen for pragmatic reasons, and as \cite{sims_robins-wasserman_2012} notes, such priors ``can unintentionally imply dogmatic beliefs about parameters of interest'' when expanded ``unthinkingly to high dimensions.''
\cite{linero_nonparametric_2023} shows that the causal inference problem we study here has this problem in spades. 
It is common, and convenient in practice, to place independent priors on $\gamma$ and $\beta$, the coefficients on $X$ in the ``propensity score'' and ``outcome'' regressions, respectively.
In low dimensions, this is innocuous. 
In high dimensions, on the other hand, it implies a \emph{dogmatic} prior belief that there is no selection bias.
If we truly believed this, of course, there would be no reason to adjust for $X$ in the first place!

To avoid prior dogmatism without the need to elicit a dependent prior for $(\gamma, \beta)$, BDML replaces the outcome regression of $Y$ on $(D, X)$ with a \emph{reduced form regression} of $Y$ on $X$.
This leaves us with a bivariate, seemingly unrelated regression of $(Y,D)$ on $X$.
Under the selection-on-observables assumption, the error terms of these two equations are correlated, and their covariance matrix contains all the information needed to carry out estimation and inference for the causal effect of $D$ on $Y$.
We show that placing independent priors on the coefficients of the seemingly unrelated regression equations runs no risk of prior dogmatism, regardless of the dimension of $X$.
Any standard prior on the error covariance matrix implies a plausible range for the \emph{a priori} selection bias.

We call our method BDML to highlight its connection to the Frequentist Double Machine Learning (FDML) approach that has become extremely popular in recent years \citep[see for example][]{chernozhukov_doubledebiased_2018}.
Like our proposal, FDML can be viewed as a means of avoiding RIC when carrying out causal inference in high dimensions.
In the partially linear regression model we study here, FDML likewise considers a bivariate reduced form regression of $(Y,D)$ on $X$.
After estimating the coefficients of these two regressions using an ML approach, FDML recovers the causal effect of interest using a further regression of the $Y$-residuals on the $D$-residuals.
This approach is computationally convenient and can be given a rigorous large-sample justification.
In practice, however, it amounts to treating the reduced form coefficients as if they were known--that is, FDML \emph{profiles} out the nuisance parameters by plugging in point estimates.
When the dimension of $X$ is large relative to sample size there may be considerable sampling uncertainty in the reduced form coefficients.
BDML explicitly acknowledges this by \emph{marginalizing} over all sources of uncertainty in the problem, leading to improved small-sample performance as measured by traditional Frequentist metrics, e.g.\ root mean-squared error (RMSE) and coverage probability. 

Crucially, the finite-sample advantages of BDML do not come at the expense of large sample guarantees.
To show this, we study a fixed coefficient asymptotic sequence in which the number of control regressors $p$ grows with sample size $n$, matching the approach of the FDML literature.
We consider three estimators: BDML, FDML, and a ``n\"{a}ive'' Bayesian alterantive that ignores RIC.
Under regularity conditions, we show that all three yield consistent estimators of the causal effect of interest.
When they are $\sqrt{n}$-consistency, we likewise show that all three have the same asymptotic variance. 
The asymptotic bias of BDML and FDML, however, is strictly smaller than that of the n\"{a}ive Bayesian approach: $p^2/n^2$ compared to $p/n$.
Moreover, the n\"{a}ive estimator is only $\sqrt{n}$-consistent if the number of controls is much smaller than sample size: $p/\sqrt{n} \rightarrow 0$.
BDML and FDML, in contrast, allow for a much larger number of controls, requiring only $p/(n^{3/4})\rightarrow 0$.
Asymptotically, BDML improves on the n\"{a}ive Bayesian approach while matching the performance of FDML. 

Continuing our asymptotic analysis, we build on results from \cite{walker_parametrization_2025} to establish a Bernstein-von Mises (BvM) theorem for BDML when $(\log n)^2 \prec p \prec \sqrt{n}$.\footnote{We use the symbol $a \prec b$ to indicate that $\lim_{n\rightarrow \infty} a/b \rightarrow 0$.}
This result further strengthens the Frequentist appeal of BDML, showing that posterior credible sets constructed using our approach are valid Frequentist confidence intervals.
Our BvM result further establishes the semi-parametric efficiency of BDML while imposing weaker assumptions on the prior than the recent Bayesian literature  
\citep[e.g.][]{luo_semiparametric_2023,breunig_semiparametric_2024}.
While we prove the BvM theorem under the assumption of Gaussian errors, these results are robust to mis-specification of the error distribution.

We conclude by comparing the finite-sample performance of BDML against a range of Bayesian and Frequentist competitors in a comprehensive simulation study and a number of empirical examples drawn from the recent literature in applied microeconomics.
Our results show that BDML provides the ``best of both worlds'', matching the asymptotic properties of Frequentist alternatives while delivering dramatically improved finite-sample performance. 
Compared existing Bayesian alternatives \citep[e.g.][]{hahn_bayesian_2020,linero_nonparametric_2023}, BDML maintains its impressive performance across a much wider range of data-generating processes (DGPs).
Since BDML is fully-Bayesian, it correctly accounts for all sources of sampling uncertainty, smoothly adapting to key features of the DGP.

This paper relates to a large recent literature on (Frequentist) double machine learning (FDML), also known as ``double-debiased machine learning''.
See \cite{ahrens_introduction_2025} for a comprehensive practical introduction and \cite{chernozhukov_doubledebiased_2018} for further theoretical details. 
FDML provides a general framework for estimating a low-dimensional target parameter using moment conditions that depend on a high-dimensional nuisance parameter. 
The na\"ive plug-in approach--constructing a preliminary ML estimate of the nuisance parameter and plugging it into the original moment condition--has two serious shortcomings.
The first is regularization bias, a more general term for RIC. 
The second is overfitting bias from using the same data twice--first to estimate the nuisance parameter and then to estimate the target parameter.
By \emph{orthogonalizing} the second-step moment condition with respect to the nuisance parameter, FDML inoculates against regularization bias to first order.
By using sample splitting, it likewise avoids overfitting bias.

Two key advantages of FDML are its simplicity and generality: whenever we can derive the appropriate orthogonalized moment condition, FDML justifies treating the estimated nuisance parameter as if it were known while providing theoretical guarantees across a wide range of first-step ML estimators.
These attractive properties have led applied researchers in economics to employ FDML in a range of empirical setting, especially the ``post-double selection'' (PDS) LASSO approach of \cite{belloni_inference_2014}.
Recent papers that apply FDML as their primary method in an observational setting include \cite{appl_dube_monopsony_2020} who estimate monopsony power in online labor markets, and \cite{appl_yang_bign_audit_2020} who estimate the ``Big N'' audit quality effect. 
While its primary motivation is to reduce bias in observational settings, FDML is also widely applied in experimental settings to flexibly incorporate covariates, either as a robustness check \citep{appl_alsan_representation_clintrials_2023, appl_armand_resource_curse_2020} or with the aim of improving estimation precision \citep{appl_hussam_psych_refugees_2022,appl_cotti_ecig_flavor_bans_2025,appl_delfino_pink_collar_2024, appl_chioda_skillstraining_uganda_2021}.
FDML has also been used to evaluate, and ultimately question, the plausibility of the selection-on-observables assumption, either by comparing against an instrumental variables (IV) benchmark \citep{appl_duflo_ghana_2017} or by creating a ``simulated'' observational dataset from a collection of randomized controlled trials \citep{appl_gordon_ms_advertising_measure_2023}. 

Although FDML is indisputably a valuable tool in applied economics, some recent work suggests that its finite sample performance can be sensitive to details of the first-step ML procedure in ways that may lead it to underperform in practice \citep{angrist_machine_2022, wuthrich_omitted_2023, bach_hyperparameter_2024}.
\cite{ahrens_introduction_2025} in particular stress that ``poorly tuned or ill-suited nuisance estimators can produce severely misleading point estimates''.
As such, they encourage researchers using FDML to ``carefully choose and tune their nuisance function estimators and consider a rich set of estimation approaches.''
While by no means an insurmountable challenge, the need for careful tuning makes FDML more complicated to implement in practice than its theoretical simplicity may suggest.
In our view, this challenge makes Bayesian alternatives to FDML particularly appealing.
By marginalizing over uncertainty in the nuisance parameter, BDML provides a fuller accounting for finite-sample estimation uncertainty.
Moreover, our method can provide robust performance across a wider range of data generating processes through the use of hierarchical priors, which adapt the degree of shrinkage to the data at hand. 

Our work also relates to a growing literature on Bayesian approaches to causal inference: see \cite{heckman_treatment_2014}, \cite{li_bayesian_2023}, and \cite{linero_how_2023} for broad discussions of this field. 
The studies closest to ours are \cite{hahn_regularization_2018,hahn_bayesian_2020} and \cite{linero_nonparametric_2023}, who consider the same partially linear model that we analyze below.\footnote{See \cite{belloni_inference_2014} for a frequentist analysis of the partially linear model.}
Compared to BDML, which marginalizes over all sources of uncertainty, these papers adopt a ``hybrid'' strategy: profiling over some parts of the problem and marginalizing over others.
This hybrid approach is computationally convenient and can be given a sound Bayesian justification.
Our simulation studies reveal, however, that the full marginalization provided by BDML is crucial for obtaining inferences that have good frequentist properties across a wide range of data generating properties. 
As mentioned above, our BvM results draw heavily on \cite{walker_parametrization_2025}.
Recent Bayesian work exploring issues related to those considered here, albeit in different models, includes \cite{saarela_bayesian_2016}, \cite{ray_semiparametric_2020}, \cite{antonelli_causal_2022},\cite{luo_semiparametric_2023}, \cite{breunig_semiparametric_2024} and \cite{breunig_double_2025}.

The remainder of this paper is organized as follows.
\Cref{sec:model} introduces our model and notation while \Cref{sec:RIC} discusses the problem of RIC in detail. 
\Cref{sec:BDML} presents our Bayesian Double Machine Learning approach and describes how it relates to both Bayesian and Frequentist alternatives, while \Cref{sec:asymptotics} presents a comparison of the asymptotic properties of BDML, FDML, and a na\"{i}ve Bayesian alternative. 
Finally, \Cref{sec:sim} presents a simulation study comparing the finite-sample performance of BDML against a range of Bayesian and Frequentist competitors.
Proofs appear in \Cref{append:proofs}.

\section{The Model}
\label{sec:model}

We consider the problem of estimating and carrying out inference for the parameter $\alpha$ in a partially linear structural model of the form 
\begin{equation}
  Y_i = \alpha D + g(X_i) + \varepsilon_i, \quad \mathbbm{E}[\varepsilon_i|D_i,X_i] = 0
\label{eq:plmY}
\end{equation}
where $Y_i$ is the outcome of interest, $D_i$ is the ``treatment'' or ``policy'' variable, and $X_i$ is a vector of control variables.\footnote{Our model and parameter of interest match those of \cite{belloni_inference_2014}.}
We do not assume that $D_i$ is binary; if it is, \eqref{eq:plmY} can be viewed as a special case of the ``selection on observables'' framework.
For reasons that will become clear below, we introduce an additional \emph{reduced form} model, namely
\begin{equation}
  D_i = m(X_i) + V_i, \quad \mathbbm{E}[V_i|X_i] = 0
\label{eq:plmD}
\end{equation}
where $m(X_i) \equiv \mathbbm{E}[D_i|X_i]$ so that \eqref{eq:plmD} holds by construction and we have
\begin{equation}
  \text{Cov}(\varepsilon_i, V_i) = \mathbbm{E}[\varepsilon V_i] = \mathbbm{E}[\varepsilon_i D_i] - \mathbbm{E}[\varepsilon_i m(X_i)] = 0
  \label{eq:CovEpsilonV}
\end{equation}
by iterated expectations and $\mathbbm{E}[\varepsilon|X_i,D_i]=0$, since $\varepsilon_i$ and $V_i$ are mean zero.

A leading example of \eqref{eq:plmY}--\eqref{eq:CovEpsilonV} assumes that $g(X_i) = X_i'\beta$ and $m(X_i)=X_i'\gamma$ where $\beta$ and $\gamma$ are $p$-vectors of unknown regression coefficients.
This simplification gives 
\begin{align}
  \label{eq:linearY}
Y_i &= \alpha D_i + X_i'\beta + \varepsilon_i, \quad \mathbbm{E}[\varepsilon_i|X_i, D_i] = 0\\ 
  \label{eq:linearD}
  D_i &= X_i'\gamma + V_i, \quad \mathbb{E}[V_i|X_i]=0
\end{align}
where we have 
\begin{equation}
  \text{Cov}(\varepsilon_i, V_i) = \text{Cov}(\varepsilon_i, D_i - X_i'\gamma) = \text{Cov}(\varepsilon_i, D_i) - \text{Cov}(\varepsilon_i, X_i')\gamma = 0.
  \label{eq:CovEpsilonVLinear}
\end{equation}

While the methods we develop in this paper are easily extended to the general case of \eqref{eq:plmY}--\eqref{eq:CovEpsilonV}, we focus on the special case of \eqref{eq:linearY}--\eqref{eq:CovEpsilonVLinear} to simplify the exposition.
We assume throughout that the dimension $p$ of $X_i$ is \emph{large relative to sample size} $n$.
If the dimension of $X_i$ exceeds the number of observations, then there is no unique ordinary least squares (OLS) estimator of $\alpha$ based on \eqref{eq:linearY}.  
Even if there are more observations than regressors, the OLS estimator of $\alpha$ is likely to be extremely noisy if $X_i$ if contains many ``relatively unimportant'' control variables--variables that, taken individually, have little predictive power for $D_i$ and $Y_i$.
This situation is arguably quite common in practice.
Applied researchers may know which controls it would \emph{suffice} to adjust for without being sure of which are most important.

\section{Regularization-Induced Confounding (RIC)}
\label{sec:RIC}
As is well-known, the OLS estimator for $\alpha$ corresponds to the Bayesian posterior mean under \eqref{eq:linearY} with a normal distribution for $\varepsilon_i$ and flat priors on $(\alpha, \beta)$.
Accordingly, in a situation where we expect OLS to perform poorly because the dimension of $X_i$ is large, a natural idea would be to ``shrink'' our estimates by introducing an informative prior on $\beta$. 
While there are many possibilities for the choice of shrinkage prior, the message of this section is that \emph{none of them} is likely to perform well in practice.
To illustrate this point, suppose we retain a flat prior for $\alpha$ but place independent normal priors on the elements of $\beta$:
\begin{equation}
  \beta^{(1)}\dots \beta^{(p)} \sim \text{iid Normal}(0, \sigma_\beta^2).
  \label{eq:ShrinkagePrior}
\end{equation}
Under the simplifying assumption that $\sigma^2_\varepsilon$ is known, the Bayes estimator for $(\alpha, \beta)$ under squared error loss has a closed-form solution:
\begin{equation}
  \begin{bmatrix}
    \widehat{\alpha}_\lambda \\
    \widehat{\beta}_\lambda
  \end{bmatrix} =
  \left[ 
    \begin{pmatrix}
    D'D & D'X \\
    X'D & X'X
\end{pmatrix} + 
\begin{pmatrix}
  0 & 0_p' \\
0_p & \lambda \mathbb{I}_p
\end{pmatrix}\right]^{-1}\begin{pmatrix}
  D'Y\\
  X'Y
\end{pmatrix}, \quad \lambda \equiv \frac{\sigma^2_\varepsilon}{\sigma^2_\beta}.
\label{eq:ridge}
\end{equation}
As seen from \eqref{eq:ridge}, the Bayes estimator in this setting corresponds to a modified version of Ridge Regression in which $\alpha$ is unpenalized.
As usual in Ridge Regression, we implicitly assume that $D_i$, $X_i$ and $Y_i$ have all been de-meaned so that the regression requires no intercept.
If \eqref{eq:ShrinkagePrior} represents genuine subjective researcher beliefs over $\beta$ then, from a Bayesian point of view, $\widehat{\alpha}_\lambda$ is the mean-squared error optimal estimator of $\alpha$.
In practice, however, shrinkage priors like \eqref{eq:ShrinkagePrior} are often used for pragmatic reasons: when $X_i$ is high-dimensional, it may be difficult if not impossible to elicit a fully-informative prior over $\beta$.
For the purposes of \emph{predicting} $Y_i$ this is usually innocuous, provided that $\sigma_\beta^2$ is chosen in a reasonable way.
But when the goal is to understand the \emph{causal effect} of $D_i$, the situation changes dramatically. 
To illustrate this point, we derive the (Frequentist) bias and variance of $\widehat{\alpha}_\lambda$ as follows under the simplifying assumption that $\varepsilon$ is homoskedastic.\footnote{\cite{hahn_regularization_2018} derive a special case of $\text{Bias}(\widehat{\alpha}_\lambda)$, but their expression contains a minor error.}

\begin{assump}
  \label{assump:ridgeMSE}
  $Y = \alpha D + X'\beta + \varepsilon$ where $\mathbbm{E}[\varepsilon|X, D] = 0$ and $\text{Var}(\varepsilon|X, D) = \sigma^2 \mathbbm{I}_p$.
\end{assump}

\begin{pro}
  \label{pro:ridgeMSE}
Let $\widehat{\rho}'\equiv (D'D)^{-1}D'X$, and $R \equiv \widehat{
\xi}'\widehat{\xi}$ where $\widehat{\xi} \equiv \left[\mathbbm{I}_p - D(D'D)^{-1}D'\right]X$.
Then \autoref{assump:ridgeMSE} implies
\begin{align}
  \label{eq:RidgeBias}
  \text{Bias}(\widehat{\alpha}_\lambda|X, D) &= \widehat{\rho}' \left[ \mathbbm{I}_p - (R + \lambda \mathbbm{I}_p)^{-1}R \right] \beta\\
  \label{eq:RidgeVariance}
  \text{Var}(\widehat{\alpha}_{\lambda}|X, D) &= \sigma^2_\varepsilon \left[ (D'D)^{-1} + \widehat{\rho}'(R + \lambda \mathbbm{I}_p)^{-1} R(R + \lambda \mathbbm{I}_p)^{-1} \widehat{\rho}\right]
\end{align}
where $\widehat{\alpha}_\lambda$ is as defined in \eqref{eq:ridge}.
\end{pro}

The $j$\textsuperscript{th} element of $\widehat{\rho}$ from \autoref{pro:ridgeMSE} is the slope coefficient from a regression of $X_i^{(j)}$ on $D_i$ without an intercept (both have been de-meaned) while the  $j$\textsuperscript{th} column of $\widehat{\xi}$ is the corresponding vector of residuals.
Setting $\lambda = 0$ corresponds to $\sigma_\beta^2  \rightarrow \infty$ and gives the OLS estimator of $\alpha$.
In this case the bias is zero and the variance simplifies to 
\[
  \text{Var}(\widehat{\alpha}_{\text{OLS}}) = \sigma^2_\varepsilon \left[ (D'D)^{-1} + \widehat{\rho}'R^{-1} \widehat{\rho}\right].
\]
If $\sigma^2_\beta$ is finite, then $\lambda \neq 0$ and $\widehat{\alpha}_\lambda$ is biased unless $\widehat{\rho}$ or $\beta$ equals zero.
Of course this in and of itself is not a convincing argument against shrinkage.  
In general, higher values of $\lambda$ imply lower values of $\text{Var}(\widehat{\alpha})$.
Even if $n > p$ so that OLS remains feasible, we should still prefer to set $\lambda > 0$ provided that the decrease in variance outweighs the increase in (squared) bias. 

The problem isn't the presence of bias but rather its dependence on other aspects of the model. 
We include covariates on the right-hand side of \eqref{eq:linearY} precisely because we believe that $\alpha$ would lack a causal interpretation without them.
This implies an \emph{a priori} belief that the control variables covary with both the outcome and treatment.
As seen from \eqref{eq:RidgeBias}, for a fixed value of $\lambda$ the bias of $\widehat{\alpha}_\lambda$ involves both $\beta$ and $\widehat{\rho}$. 
All else equal, larger values of $\widehat{\rho}^{(j)}\beta^{(j)}$ imply a higher bias.
Our beliefs about the magnitude of $\beta^{(j)}$ are already accounted for, at least in principle.
Since $\lambda = \sigma_\varepsilon^2 / \sigma_\beta^2$, a belief that $|\beta^{(j)}|$ will be large on average leads us to shrink \emph{less}, using a smaller value of $\lambda$ and keeping the bias in check. 
But $\text{Bias}(\widehat{\alpha}_\lambda)$ also depends crucially on $\widehat{\rho}$, and $\lambda$ fails to take this into account.\footnote{The stated priors on $(\alpha, \beta)$ along with the fixed known value of $\sigma_\varepsilon^2$ \emph{implicitly} encode beliefs about the likely magnitude of $\widehat{\rho}$, as we explain further below.}
If the treatment and controls are strongly correlated, even a reasonable choice for $\sigma_\beta^2$ can produce a highly biased estimator of $\alpha$, a phenomenon that \cite{hahn_regularization_2018} call \textbf{regularization-induced confounding (RIC)}. 
The name alludes to the fact that ``regularizing'' the OLS estimator by setting $\lambda > 0$ causes the residuals to violate the sample analogue of $\mathbbm{E}[D_i\varepsilon_i|X_i]=0$, the population moment condition that gives $\alpha$ a causal interpretation. 
While RIC is typically equated with unacceptably high bias, numerical experiments using \eqref{eq:RidgeBias}--\eqref{eq:RidgeVariance} confirm that RIC implies an unfavorable bias-variance \emph{trade-off}.
This is the sense in which we use the term. 

The trouble, it seems, is failing to account for our prior beliefs about $\gamma$, the reduced form coefficient from \eqref{eq:linearD} that reflects the dependence between $D_i$ and $X_i$.
A natural solution, then, would be to incorporate these beliefs into our model.
Doing so requires us to add an equation and work with a \emph{bivariate} regression including both \eqref{eq:linearY} and \eqref{eq:linearD}.
Staying as close as possible to the single-equation example from above, suppose we specify a normal likelihood for $(\varepsilon_i, V_i)$ and place independent normal shrinkage priors on the elements of $\beta$ and $\gamma$, all centered at zero.
In this model, an additional hyperparameter $\sigma_\gamma^2$ quantifies our prior beliefs about the strength of dependence between the treatment and control variables.
Unfortunately, $\sigma_\gamma^2$ will have \emph{no effect whatsoever} on our estimates or inferences for $\alpha$. 
Recall from \eqref{eq:CovEpsilonVLinear} above that the causal assumptions of our model require $\text{Cov}(\varepsilon_i, V_i) = 0$.
Adopting the shorthand $\theta$ to refer to all model parameters, this implies that the likelihood factorizes as follows
\begin{equation}
f(Y, D|X, \theta)= f(Y|D,X,\theta)f(D|X,\theta) = f(Y|D,X, \alpha, \beta, \sigma_\varepsilon^2) f(D|X,\gamma, \sigma_V^2).
  \label{eq:LikeFactorize}
\end{equation}
Notice that the first factor involves only $(\alpha, \beta, \sigma_\varepsilon^2)$ while the second involves only $(\gamma, \sigma_V^2)$.
Therefore, unless we specify a prior in which $\beta$ and $\gamma$
Treating $\sigma_\varepsilon^2$ as known, the bias and variance expressions from \eqref{eq:RidgeBias}--\eqref{eq:RidgeVariance} continue to hold and the problem of RIC remains. 
\cite{linero_nonparametric_2023} calls a causal model in which the likelihood factorizes as in \eqref{eq:LikeFactorize} and $\beta$ and $\gamma$ are \emph{a priori} independent a \textbf{Bayesian-ignorable} model, and points out that such models imply a highly dogmatic prior for the extent of selection bias when $p$ is large. 
This is another lens through which to view RIC: it arises from ``accidentally'' imposing a strong prior belief that there is little or no selection bias.
We explore this phenomenon in more detail in \autoref{sec:prior} below.

%!TEX root = ./main.tex

\section{Bayesian Double Machine Learning (BDML)}
\label{sec:BDML}

In principle, specifying an informative prior that incorporates dependence between $\beta$ from \eqref{eq:linearY} and $\gamma$ from \eqref{eq:linearD} would solve the problem of RIC introduced in \autoref{sec:RIC}.
Eliciting such a prior in practice, however, may be extremely challenging, especially when the number of control variables is large.
This motivates us to introduce an alternative approach, one that is inspired by ideas from the recent Frequentist literature on so-called ``double-debiased machine learning'', often referred to as ``double machine learning'' for short.
Substituting \eqref{eq:linearD} into \eqref{eq:linearY} gives the reduced form regression
%\begin{equation}
%  Y = \boldsymbol{X}' (\alpha \boldsymbol{\gamma} + \boldsymbol{\beta}) + (\varepsilon + \alpha V) = \boldsymbol{X}'\boldsymbol{\delta} + U
%  \label{eq:linearYrf}
%\end{equation}
\begin{equation}
  Y_i = X_i' (\alpha \gamma + \beta) + (\varepsilon_i + \alpha V_i) = X_i'\delta + U_i
  \label{eq:linearYrf}
\end{equation}
where $U_i \equiv \varepsilon_i + \alpha V_i$ and $\delta \equiv \alpha \gamma + \beta$.
Since since $\mathbbm{E}[\varepsilon_i|X_i,D_i]=0$ by \eqref{eq:linearY} and $X_i$ is uncorrelated with $V_i$ by \eqref{eq:linearD}, $X_i$ is uncorrelated with $U_i$. 
Assuming normal errors for simplicity, this gives the following \emph{bivariate} reduced form regression model
\begin{gather}
  \begin{aligned}
    Y_i &= X_i'\delta + U_i \\
    D_i &= X_i'\gamma + V_i
  \end{aligned} 
  \qquad
  \left.
  \begin{bmatrix}
    U_i \\ V_i
  \end{bmatrix} \right| X_i \sim \text{Normal}(\mathbf{0}, \Sigma).
  \label{eq:RFreg}
\end{gather}
Given that our goal is to estimate and carry out inference for $\alpha$, it may seem odd that we have eliminated this parameter from the regression equation for $Y$.
Crucially, however, the errors $(U_i, V_i)$ in \eqref{eq:RFreg} are \emph{correlated} whenever $\alpha\neq 0$.
This is what will allow us to learn the causal effect of interest.
By \eqref{eq:CovEpsilonVLinear} we have $\text{Cov}(\varepsilon,V) = 0$ and thus
\begin{align*}
  \text{Var}(U_i) &=  \text{Var}(\varepsilon_i + \alpha V_i) = \text{Var}(U_i) + \alpha^2 \text{Var}(V_i) = \sigma_\varepsilon^2 + \alpha^2 \sigma_V^2 \\
  \text{Cov}(U_i,V_i) &=  \text{Cov}(\varepsilon_i + \alpha V_i, V_i) = \alpha \sigma_V^2
\end{align*}
by the definition of $U$.
Therefore, 
\begin{equation}
  \Sigma \equiv \text{Var}\begin{bmatrix}
    U_i\\ V_i
  \end{bmatrix} = \begin{bmatrix}
    \sigma_U^2 & \sigma_{UV} \\
    \sigma_{UV} & \sigma_V^2
  \end{bmatrix} =
  \begin{bmatrix}
    \sigma_{\varepsilon}^2 + \alpha^2 \sigma_V^2 & \alpha \sigma_V^2 \\
    \alpha \sigma_V^2 & \sigma_V^2
  \end{bmatrix}.
  \label{eq:Sigma}
\end{equation}
Notice that knowledge of $\Sigma$ immediately implies knowledge of $\alpha$ via the simple transformation 
\begin{equation}
  \frac{\text{Cov}(U,V)}{\text{Var}(V)} = \frac{\sigma_{UV}}{\sigma_V^2} =  \frac{\alpha \sigma_V^2}{\sigma_V^2} = \alpha.
  \label{eq:alphaFromSigma}
\end{equation}
The relationship in \eqref{eq:alphaFromSigma} suggests the following approach to estimation and inference for $\alpha$.

\begin{alg}[Bayesian Double Machine Learning]
  \label{alg:linearBDML}
  \mbox{}
\begin{enumerate}
  \item Place priors on $(\delta, \gamma, \Sigma)$ in the reduced form model given by \eqref{eq:RFreg}.
  \item Sample from the joint posterior $(\delta, \gamma, \Sigma)|(X,D,Y)$ where $\Sigma$ is defined in \eqref{eq:Sigma}.
    \item Transform the posterior draws for $\Sigma$ to obtain a posterior for $\alpha = \sigma_{UV}/\sigma_V^2$.
\end{enumerate}
\end{alg}
The procedure given in \autoref{alg:linearBDML} is simple, flexible, and fully Bayesian. 
At the same time, it allows us to avoid the problem of inadvertently specifying a dogmatic prior on selection bias, pointed out by \cite{linero_nonparametric_2023} and explained in \autoref{sec:RIC} above.
First, because we work with a pair of reduced form regressions rather than a structural regression for $Y_i$ and a reduced form regression for $D_i$, our likelihood does not factorize: $(U_i,V_i)$ are correlated whereas $(\varepsilon_i,V_i)$ are not.
Second, because $\delta = (\alpha \gamma + \beta)$, placing independent priors on the reduced form regression coefficients $\delta$ and $\gamma$ does \emph{not} imply independent priors for $\beta$ and $\gamma$.  
Hence, \eqref{eq:RFreg} in general satisfies \emph{neither} of the two conditions required for Bayesian-ignorability introduced by \cite{linero_nonparametric_2023}.
Naturally, the properties of point and interval estimates based on \autoref{alg:linearBDML} will depend on the specific choice of priors for $(\delta, \gamma,  \Sigma)$.
We explore this point in detail below and and compare the theoretical and practical performance of different priors.
The key insight of our approach, however, is that re-writing the model from \eqref{eq:linearY}--\eqref{eq:CovEpsilonVLinear} in the form given by \eqref{eq:RFreg} allows us to side-step the problem of RIC using a standard Bayesian multivariate regression model with conditionally conjugate priors.
As far as we are aware, this observation is new to the literature. 

Because the estimator $\widehat{\alpha}_\lambda$ from \eqref{eq:ridge} can be viewed from a Frequentist perspective as an instance of Ridge Regression, \eqref{eq:RidgeBias}--\eqref{eq:RidgeVariance} imply that RIC is just as much of a challenge for frequentists who apply machine learning methods to high-dimensional causal inference problems as it is for Bayesians.
The increasing recognition of this problem over the past decade has led to the development of a framework called (frequentist) double machine learning (FDML).
See \cite{chernozhukov_doubledebiased_2018} for an overview.
The simplest FDML approach to estimating $\alpha$ in \eqref{eq:linearY} proceeds by running three regressions.
First, let $\hat{\delta}_{\text{ML}}$ be the estimated coefficient vector from some unspecified ``machine learning'' (i.e.\ regularized) regression of $Y_i$ on $X_i$, e.g.\ ridge regression, LASSO, random forests, etc.\
Second, let $\hat{\gamma}_{\text{ML}}$ be the estimated coefficient vector from some unspecified regularized regression of $D_i$ on $X_i$.
The FDML estimator is then obtained by regressing the residuals from the $Y_i$ on $X_i$ regression on the residuals from the $D$ on $X$ regression, in other words
%\begin{equation}
%  \widehat{\alpha}_{\text{FDML}} = \frac{\sum_{i=1}^n (Y_i - \boldsymbol{X}_i'\widehat{\boldsymbol{\delta}}_\text{ML} )(D_i - \boldsymbol{X}_i'\widehat{\boldsymbol{\gamma}}_\text{ML} )}{\sum_{i=1}^n (D_i - \boldsymbol{X}_i'\widehat{\boldsymbol{\gamma}}_\text{ML})^2}.
%  \label{eq:FDML}
%\end{equation}
\begin{equation}
  \widehat{\alpha}_{\text{FDML}} = \frac{\sum_{i=1}^n (Y_i - X_i'\widehat{\delta}_\text{ML} )(D_i - X_i'\widehat{\gamma}_\text{ML} )}{\sum_{i=1}^n (D_i - X_i'\widehat{\gamma}_\text{ML})^2}.
  \label{eq:FDML}
\end{equation}
Like \eqref{eq:RFreg} from above, \eqref{eq:FDML} involves reduced form regressions of $Y$ and $D$ on $X_i$. 
Because of this similarity, we refer to the approach from \autoref{alg:linearBDML} as \textbf{Bayesian Double Machine Learning (BDML)}.
From a Bayesian perspective, BDML provides a way to guard against RIC while using a fully generative model and conducting inference in a manner that respects the likelihood principle.
This gives researchers access to the full arsenal of Bayesian tools, from hierarchical modeling to posterior predictive checks. 
But even for researchers who are primarily concerned with frequentist performance, BDML offers two key advantages. 
First, it allows practitioners to incorporate subject matter expertise into estimation through the choice of prior, potentially yielding sizeable efficiency gains in small samples.
To show how this can be achieved in practice, we propose a simple and effective way of incorporating researcher beliefs concerning the R-squared of the reduced form regressions, building on the so-called ``R2D2'' prior of \citep{zhang_bayesian_2022}.
Second, rather than simply ``plugging in'' point estimates of $\delta$ and $\gamma$, as in FDML, our BDML approach marginalizes over these high-dimensional parameters, correctly accounting for all sources of estimation uncertainty for $\alpha$.
As such, BDML may provide more accurate frequentist inference while avoiding the need for cumbersome and computationally-expensive ``cross-fitting'' procedures, as are commonly employed in applications of FDML.

While ours is not the only Bayesian proposal for avoiding RIC in a model like \eqref{eq:RidgeVariance}, BDML differs in important ways from existing methods.  
The two closest to our approach are the methods of \cite{hahn_regularization_2018,hahn_bayesian_2020} on the one hand and \cite{linero_nonparametric_2023}, on the other hand.
Both begin with a preliminary Bayesian linear regression of $D_i$ on $X_i$ to yield a Bayes estimator of $\gamma$ that is then substituted into a second Bayesian regression.\footnote{While \cite{hahn_regularization_2018} takes a more involved joint estimation approach, their followup paper \cite{hahn_bayesian_2020} advocates the simpler two-step approach that we describe here.}
The approaches differ in the precise regression model that they estimate in the second step.
\cite{hahn_regularization_2018} rely on the same substitution as our \eqref{eq:linearYrf} but, rather than absorbing $\alpha V_i$ into a reduced form error term, replace it with its point estimate $\widehat{V}_i$ based on a first-step regression.
\begin{alg}[HCPH]
  \label{alg:hahn} 
Obtain the marginal posterior for $\alpha$ from the linear regression 
\[
  Y_i = \alpha \widehat{V}_i + X_i'\delta_2 + \nu_i, \quad
  \widehat{V}_i \equiv (D_i - X_i' \widehat{\gamma})
\]
where $\widehat{\gamma}$ is a first-step Bayes-estimator of $\gamma$ from \eqref{eq:linearD}. 
\end{alg}
Note that the parameter $\delta_2=\delta+\alpha \widehat{\gamma} $ from \autoref{alg:hahn} does \emph{not} in general coincide with our $\delta$ from \eqref{eq:linearYrf}.
\cite{linero_nonparametric_2023} generalizes the second-stage regression from \cite{hahn_regularization_2018}. 
Defining $\widehat{D}_i \equiv  X_i'\widehat{\gamma}$, we can re-express the regression from \autoref{alg:hahn} as
\[
  Y_i = \alpha D_i - \alpha \widehat{D}_i + X_i'\delta_2 + \nu_i. 
\]
\cite{linero_nonparametric_2023} estimates the parameters of this regression \emph{without the restriction} that the coefficient on $\widehat{D}_i$ equals the negative of the coefficient on $\widehat{D}_i$.

\begin{alg}[Linero]
  \label{alg:linero} 
Obtain the marginal posterior for $\alpha$ from the linear regression
\[
  Y_i = \alpha D_i + \kappa \widehat{D}_i + X_i'\delta_3 + \nu_i, \quad \widehat{D}_i \equiv X_i' \widehat{\gamma}
\]
where $\widehat{\gamma}$ is a first-step Bayes-estimator of $\gamma$ from \eqref{eq:linearD}. 
\end{alg}

Notice that the regression from \autoref{alg:linero} has an identification problem: $\widehat{D_i} = X_i\widehat{\gamma}$ is \emph{perfectly collinear} with $X_i$.
We write $\delta_3$ as the vector of regression coefficients for $X_i$ to emphasize that, when estimated \emph{without} the restriction that $\kappa = -\alpha$, this procedure will not coincide with \autoref{alg:hahn}.

Compared to the procedures in Algorithms \ref{alg:hahn}--\ref{alg:linero}, our BDML approach differs in two important ways. 
First, our inferences for $\alpha$ incorporate posterior uncertainty over \emph{both} $\gamma$ and $\delta$.
\cite{hahn_bayesian_2020} argue that, since Bayesian inference in a regression model \emph{conditions} on the realizations of the regressors, a two-step approach can be viewed as fully-Bayesian: any Bayes-estimator of $\gamma$ from the first-step regression is necessarily a measurable function of $(D, X)$, so plugging this estimate into a second-step regression does not violate the likelihood principle.
We agree with this argument, as far as it goes. 
But there are \emph{countless} approaches that adhere to the likelihood principle, all of which yield different posteriors. 
As we show in our simulations below, accounting for uncertainty in $\widehat{\gamma}$ is crucial for obtaining inferences that have good frequentist properties in addition to a sound Bayesian justification.
%Second, even if $(V, \varepsilon)$ from \eqref{eq:linearY}--\eqref{eq:linearD} are homoskedastic, the compound error term $\eta_i = \varepsilon + \alpha \boldsymbol{X}_i(\widehat{\boldsymbol{\gamma}} -\boldsymbol{\gamma})$ from Algorithms \ref{alg:hahn}--\ref{alg:linero} is \emph{necessarily} heteroskedastic.
%creating challenges for any flavor of inference, Bayesian or frequentist.
%In contrast, the compound error term $U = (\alpha V + \varepsilon)$ is homoskedastic whenever $V$ and $\varepsilon$ are.
Second, by residualizing \emph{both} $D_i$ and $Y_i$ with respect to $X_i$, BDML achieves robustness properties similar to those of frequentist double machine learning in a fully-Bayesian model.
In contrast \cite{hahn_regularization_2018} residualizes only $D_i$.
Alluding to similar robustness concerns raised by \cite{zigler_central_2016},  \cite{linero_nonparametric_2023} suggests that his two-step approach could give added robustness when either \eqref{eq:linearY} or \eqref{eq:linearD} may be misspecified by eliminating ``feedback'' from the estimation problem.
As we show below, however, our approach performs well in practice even in the simulation designs from \cite{linero_nonparametric_2023}.
This is because our approach, like FDML, is ``double de-biased'' where as Algorithms \ref{alg:hahn}--\ref{alg:linero} are only ``single de-biased'', a point we expand upon below.

%!TEX root = ./main.tex

\section{Asymptotic Properties}
\label{sec:asymptotics}

\subsection{Assumptions}

We now explore the large-sample properties of the BDML approach from \autoref{alg:linearBDML} and contrast them with those of the frequentist alternative, FDML, as well as those of a ``na\"{i}ve'' Bayesian approach analogous to the Ridge Regression estimator from \eqref{eq:ridge}.
We first introduce some notation.
Let $X$ be the $(n\times p)$ matrix of control regressors from above and define the $(n\times 2)$ matrix of ``outcomes'' $W$, the $(p\times 2)$ matrix of coefficients $B$, and the $(n \times 2)$ matrix of errors $E$ according to
\begin{equation}
  W \equiv \begin{bmatrix}
    Y & D
  \end{bmatrix}, \quad
  B \equiv \begin{bmatrix}
    \delta & \gamma
  \end{bmatrix}, \quad
  E \equiv \begin{bmatrix}
    U & V
  \end{bmatrix}
  \label{eq:RFmatrixDefn}
\end{equation}
where $E_i' \equiv (U_i, V_i)$ denotes a specified row of $E$. 
Then, we can express the reduced form regression model from \eqref{eq:RFreg} in matrix form as
\begin{equation}
  W = X B + E,\quad E_i|X \sim \text{ iid Normal}_2(0, \Sigma).
\label{eq:RFmatrixReg}
\end{equation}
We consider an asymptotic sequence in which the number of control regressors $p$ grows with sample size $n$ and the true parameter values are \emph{fixed}.
We write these fixed true values as $\Sigma^*$ and $B^* \equiv \begin{bmatrix} \delta^* & \gamma^* \end{bmatrix}$ to distinguish them from the random variables $\Sigma$ and $B$ that represent prior (or posterior) uncertainty in a Bayesian model.
Analogously, when we have occasion to refer to the true values of the structural parameters from \eqref{eq:linearY}, we write them as $(\alpha^*,\beta^*\sigma_\varepsilon^*)$.

Both for simplicity and to underscore the fact that BDML does not require the researcher to specify unusual or elaborate priors, we work with the classic conditionally conjugate prior for seemingly unrelated regression (SUR) models.\footnote{See \cite{zellner_introduction_1971} section 8.5 for a textbook treatment.}
In particular, we combine \eqref{eq:RFmatrixReg} with the prior $\pi(\Sigma,B) = \pi(\Sigma)\pi(\delta)\pi(\gamma)$ where
\begin{equation}
    \Sigma \sim \text{Inverse-Wishart}(\nu_0, \Sigma_0), \quad
    \delta \sim \text{Normal}_p(0,\mathbbm{I}_p/\tau_\delta), \quad
  \gamma \sim \text{Normal}_p(0, \mathbbm{I}_p/ \tau_\gamma)
  \label{eq:SURprior}
\end{equation}
and $(\tau_\delta, \tau_\gamma)$ denote the prior precisions for the reduced-form parameters $\delta$ and $\gamma$.\footnote{As we explain further below, our use of a normal likelihood for $E_i$ is innocuous.}
For purposes of comparison, we also consider a na\"ive Bayesian estimator given by the Ridge Regression estimator from \eqref{eq:ridge}.
This approximation treats the error variance $\sigma_\varepsilon^2$ as a known constant and imposes the prior $\beta \sim \text{Normal}_p(0, \mathbbm{I}_p/\tau_\beta)$, where $\lambda = \tau_\beta \sigma_\varepsilon^2$ 

%To permit a clean comparison between our BDML approach based on the reduced form regressions in \eqref{eq:RFreg} and the more familiar approach based on the structural equation \eqref{eq:linearY}, we define a ``na\"ive'' alternative that is as similar as possible to \eqref{eq:SURprior}, representing a fully-Bayesian equivalent of the Ridge Regression estimator from \eqref{eq:ridge}.
%In particular, our na\"ive model combines $\varepsilon_i|(X,D) \sim \text{iid Normal}(0, \sigma^2_\varepsilon)$ with $\pi(\sigma^2_\varepsilon, \alpha, \beta)\propto \pi(\sigma_\varepsilon^2) \pi(\beta)$ where
%\begin{equation}
%  \sigma_\varepsilon^2 \sim \text{Inverse-Gamma}(a_0, b_0), \quad
%\beta \sim \text{Normal}_p(0,\mathbbm{I}_p/\tau_\beta)
%  \label{eq:naiveprior}
%\end{equation}
%and $\tau_\beta$ denotes the prior precision for $\beta$.
%Note that the na\"ive model places a flat prior on $\alpha$

Equation \ref{eq:SURprior} describes the Bayesian \emph{model} we study.
We now describe the assumptions that we place on the true data-generating process.
For ease of notation, we state these in terms of the \text{reduced form} model.
Combining these assumptions with \eqref{eq:linearY}--\eqref{eq:linearD} implies closely related assumptions for the structural model.
The causal structure, in turn, gives a meaningful causal interpretation to $\alpha = \sigma_{UV}/\sigma_V^2$.
We first assume that the researcher observes a sample of $n$ iid observations from the reduced form model.

\begin{assump}[Random Sampling]
  \label{assump:dgp}
  The random variables $\left\{(X_i, E_i)\right\}_{i=1}^n$ are iid across $i$ and $W_i = (Y_i, D_i)$ is generated according to \eqref{eq:RFmatrixReg} with true parameters $B^*$ and $\Sigma^*$.
\end{assump}

%\todo[inline]{Do we need to add something to the preceding to refer to the structural model? Or the way that $E_i$ is related to $X_i$? It seems like we need something. More broadly, are we happy to state assumptions in terms of the reduced form?}

%\todo[inline]{Need to refer to the structural model from \eqref{eq:linearY} and \eqref{eq:linearD}. These give the interpretation of $\alpha$. But they don't actually give a likelihood. For this we need to refer to the reduced from model from above. The derivations rely on normality, but we can relax it via the BvM. But we can still place assumptions on the reduced form: it's a linear mapping!}

While \autoref{assump:dgp} imposes normality of the errors $E_i$, as we explain further below this assumption turns out to be innocuous.
Our next assumption concerns the true reduced form parameters $\Sigma^*$ and $(\delta^*, \gamma^*)$.
As mentioned above, we consider an asymptotic sequence in which the number of control regressors, $p$, grows with sample size.
This implies that the dimensions of $\delta$ and $\gamma$ likewise grow.
To accommodate this, we consider $\delta^*$ and $\gamma^*$ to be square-summable infinite \emph{sequences} of parameters, i.e.\ elements of the sequence space $\ell^2(\mathbbm{N})$. 
In finite samples only the first $p$ elements of each are ``active'', but as we obtain more control regressors, the number of regression coefficients increases to match.
Our assumption ensures that these coefficients do not ``explode'' as the number of regressors increases.
Intuitively, each additional control matters ``less on average'' the more that we include in our model. 
Because the dimension of $\Sigma^*$ does not change with sample size, we need only assume that this variance-covariance matrix is well-behaved.

%\todo[inline]{We haven't yet stated the relative rates of $p$ and $n$ so it seemed strange to write an assumption about bounded $\delta^*$ and $\gamma^*$ as $n$ grows. I think the way I've stated it here makes no difference to the proofs, but is potentially a bit clearer and more explicit about the asymptotic regime we're considering.}

\begin{assump}[True Reduced Form Parameters]
  \label{assump:trueRF}
  \mbox{}
  \begin{enumerate}[(i)]
    \item  $\Sigma^* \equiv \text{Var}(E_i)$ is finite and positive definite.
    \item $\|\delta^*\|_2$ and $\|\gamma^*\|_2$ are bounded as $p\rightarrow \infty$ 
    %\item $\lim_{p\rightarrow \infty} \sum_{j=1}^p (\delta_j^*)^2 < \infty$ and $\lim_{p\rightarrow \infty} \sum_{j=1}^p (\gamma_j^*)^2 < \infty$
  \end{enumerate}
\end{assump}

Our next assumption concerns the control regressors $X_i$.
For simplicity, and without loss of generality, we assume that $X_i$ has been centered around its mean.
To make it easier to compare our BDML procedure against its frequentist cousin, we further assume that there are more observations than controls: $p < n$. 
%This ensures that \todo{???}{covariance matrix of residuals is non-degenerate}.
The remaining conditions we impose on $X_i$ restrict the variances and covariances of the control regressors as $p$ grows.
To state them, let $\Sigma_X\equiv \text{Var}(X_i)$ and $\lambda_j(\Sigma_X)$ be the $j$-th eigenvalue of $\Sigma_X$, where these values are arranged in descending order.
Further let $F_p(t) \equiv \frac{1}{p} \sum_{j=1}^{p} \mathbbm{1}\left\{\lambda_j(\Sigma_X) \leq t\right\}$ be the empirical spectral distribution of $\Sigma_X$.

\begin{assump}[Control Regressors]
  \mbox{}
  \label{assump:controls}
  \begin{enumerate}[(i)]
    \item $\mathbbm{E}(X_i) = 0$, without loss of generality, and $p < n$.
    \item $\Sigma_X \equiv \text{Var}(X_i)$ is finite and positive definite for any $p$
    \item The empirical spectral distribution $F_p(t) \equiv \frac{1}{p} \sum_{j=1}^{p} \mathbbm{1}\left\{\lambda_j(\Sigma_X) \leq t\right\}$ converges to a limit distribution $F(t)$ supported on $(0,\infty)$ for all continuity points of $F$ as $p\rightarrow \infty$.
    \item $\frac{1}{p}\sum_{j=1}^p \left[\lambda_j(\Sigma_X)\right]^m \rightarrow \mu^{(m)} < \infty$ as $p \rightarrow \infty$ for $m = -1, 2, 2$ and $2 + \eta$ for some $\eta > 0$.
  \end{enumerate}
\end{assump}

Parts (iii) and (iv) of \autoref{assump:controls} ensure that the spectrum of $\Sigma_X$ remains ``well-behaved'' as the dimension of $X_i$ grows.
The final ingredient of our asymptotic framework is an assumption about the relative rates of $n$, $p$, and the prior precisions. %from \eqref{eq:SURprior} and \eqref{eq:naiveprior}.
To state it we use the notation $a_n \asymp b_n$ to indicate that two sequences are ``of the same order'', i.e.\ that $a_n = O(b_n)$ and $b_n = O(a_n)$. 

\begin{assump}[Rate Restrictions]
  \label{assump:rates}
  \mbox{}
  \begin{enumerate}[(i)]
    \item $p = o(n)$
    \item $\tau_\delta, \tau_\gamma, \lambda = o(n)$
    \item $\tau_\delta, \tau_\gamma, \lambda \asymp p$
  \end{enumerate}
\end{assump}

Part (i) of \autoref{assump:rates} states that the sample size dominates the number of controls and is the asymptotic equivalent of $p < n$ from \autoref{assump:controls} (i) above.\footnote{Formally \autoref{assump:rates} (i) is redundant given (ii) and (iii), but we state it explicitly to emphasize a key feature of our asymptotic framework.}
Since $p \rightarrow \infty$, this means that there are many controls but not \emph{so many} that they overwhelm the available data.
Part (ii) of \autoref{assump:rates} states that the sample size dominates each of the prior precisions.
Intuitively, this means that the priors are only ``weakly informative'' in the sense that they will be dominated by the data in large samples.
Along this asymptotic sequence, the shrinkage bias of the Ridge estimator from \eqref{eq:ridge} vanishes in the limit.
An $o(n)$ rate for the prior precisions also accords with the typical assumption from the frequentist literature on double machine learning.
Finally, part (iii) asserts that the prior precisions are of the same order as the number of controls.
This means that our prior imposes \emph{more shrinkage} when there are more controls.
Taken together, assumptions \autoref{assump:dgp}--\autoref{assump:rates} implied that, under the prior from \eqref{eq:SURprior}, the implied R-squared values from the regressions of $D$ on $X$ and $Y$ on $X$ remain well-behaved in the limit: they do not diverge to zero or one.
See \autoref{pro:rsquared} in the Appendix for details.

\subsection{Prior Properties}
\label{sec:prior}

Our BDML approach places priors on the reduced form coefficients $(\Sigma, B)$.
We now consider the prior that \eqref{eq:SURprior} \emph{implies} for other quantities of interest. 
We begin with the causal effect of interest, $\alpha$.
Partition $\Sigma$ into blocks $\Sigma_{11}, \Sigma_{12}$ and $\Sigma_{22}$ and likewise partition $\Sigma_0$ into $\Sigma_{0,11}$ and so on.
Then $\alpha = \Sigma_{12}/\Sigma_{22}$.
Under \eqref{eq:SURprior}, our prior for $\alpha$ follows a location-scale t-distribution.

\begin{lem}[Induced prior on \(\alpha\)] \label{lem:alpha-prior}
  Under \eqref{eq:RFmatrixReg} and \eqref{eq:SURprior}, the induced marginal prior on \(\alpha\) is a location-scale \(t\)-distribution, in other words
\[ \alpha \sim \frac{\Sigma_{0,12}}{\Sigma_{0,22}} + \frac{|\Sigma_0|^{1/2}}{\sqrt{\nu_0}\Sigma_{0,22}}t_{\nu_0},\] 
where $t_{\nu_0}$ denotes a Student-t random variable with $\nu_0$ degrees of freedom.
\end{lem}

If the prior scale matrix $\Sigma_0$ is diagonal, as is typical in practice, then $\Sigma_{0,12}/\Sigma_{0,22} = 0$ so the implied prior for $\alpha$ is symmetric around zero. 
Choosing a small value for $\nu_0$ gives this prior heavy tails, making it compatible with both very large positive and negative values for the effect of interest.
This is reassuring: although we have placed a prior on the reduced form rather than the structural model, we can easily avoid imposing unintentionally strong prior beliefs on $\alpha$ by making $\Sigma_0$ diagonal.

We now consider the implied prior for a different but equally important quantity, namely the \emph{selection bias}.
This idea is most easily explained in the case where $D_i$ is binary.
For the general version, see 
Propositions \ref{pro:conf-bias-naive}--\ref{pro:conf-bias-BDML} in the Appendix.
Following the usual convention in applied microeconomics, we define the selection bias as the difference between the quantity identified by a simple comparison of mean outcomes, treated minus untreated, and the true causal effect.
Under \eqref{eq:linearY}, this becomes
\[
  \text{SB} \equiv \left[ \mathbbm{E}(Y_i|D_i=1) - \mathbbm{E}(Y_i|D_i=0) \right] - \alpha = \left[ \mathbbm{E}(X_i|D_i=1) - \mathbbm{E}(X_i|D_i=0)' \right]\beta.
\]

As mentioned in \autoref{sec:RIC} above,  \cite{linero_nonparametric_2023} shows that any model that factorizes as in \eqref{eq:LikeFactorize} and in which and $\beta$ and $\gamma$ are \emph{a priori} independent implies a highly dogmatic prior for the extent of selection bias when $p$ is large.
This is another lens through which to view RIC: it arises from ``accidentally'' imposing a strong prior belief that there is little or no selection bias.
The following result shows that the na\"ive estimator suffers from this problem whereas our BDML estimator does not.

\begin{pro}[Selection bias: na\"ive estimator vs BDML] \label{pro:sel-bias}
  %Assume that the model and prior are as above. 
  Suppose that Assumptions \ref{assump:dgp}--\ref{assump:rates}  hold.
  Then, as $n,p \rightarrow \infty$ the naive (Ridge) prior implies 
  \[
  \SB_\naive=\frac{\gamma'\Sigma_X\beta}{\Sigma_{22}+\gamma'\Sigma_X\gamma}\pto0,
  \]
  whereas the prior from \eqref{eq:SURprior} implies 
  \[
  \SB_\BDML\pto\frac{\Sigma_{12}}{\Sigma_{22}+\gamma'\Sigma_X\gamma}.
  \]
\end{pro}

When $n$ and $p$ are large, the na\"ive approach in effect \emph{assumes away selection bias}: it implies a prior that is concentrated in a small neighborhood around zero.
But if there is no selection bias, then there is no need to adjust for $X_i$ in the first place: it is as if $D_i$ had been randomly assigned!
Intuitively, in high-dimensional spaces randomly chosen vectors tend to be nearly orthogonal to each other.
Since $\Sigma_X$ is positive definite, the numerator \(\gamma'\Sigma_X\beta\) from $\text{SB}_{\text{naive}}$ converges in probability to zero. 
In contrast, the numerator of $\text{SB}_{\text{BDML}}$ is $\Sigma_{12}$, which has a proper prior that is unaffected by the sample size or the number of control regressors.
This is a clear advantage of working with the reduced form: stating a prior in which $\delta$ and $\gamma$ are independent is innocuous because the implied prior for $\delta$ and $\beta = \delta - \alpha \gamma$ are \emph{not independent}.

\subsection{Posterior Properties}

Let $\beta^*=\delta^*-\alpha^*\gamma^*$. We have the following result on the posterior mean of  $\alpha$ under the na\"ive estimator.
\begin{pro}[Posterior mean: na\"ive estimator] 
  \label{pro:post-mean-naive}
  Assume that the model and prior are as above. Suppose that Assumptions \ref{assump:dgp}--\ref{assump:controls} hold. Then, the posterior mean of \(\alpha\) is given by 
  \[
  \hat\alpha_\naive=\alpha^*+\frac{\lambda} {n}\frac {\gamma^{*\prime}\beta^*} {\Sigma^*_{22}}+o\left(\frac {p} {n}\right)+O_p\left(\frac 1 {\sqrt n}\right),
  \]
  where the \(O_p(1/\sqrt n)\) term is mean zero.
  Therefore, the naive estimator is consistent, but not \(\sqrt n\)-consistent for \(\lim_{n\to\infty}p^2/ n=\infty\) and \(\gamma^{*\prime}\beta^*\neq0\).
\end{pro}

Note that a similar (in)consistency argument can be established for the cases with \(\lambda,p\asymp n\), as \eqref{eq:naive-fc} still holds.

\paragraph{Conditional posterior} Note that by Assumption \ref{assump:trueRF}, the true \(\Sigma^*\) is finite and positive definite, so \(\Sigma^{-1}\) is also finite and positive definite in the limit, almost surely in the posterior. Also by Assumption \ref{assump:controls}, \(\Sigma_X\) is finite with a positive proportion of positive eigenvalues. If \(X'X\) does not have full rank, its inverse is given by the Moore-Penrose pseudoinverse.

The conditional posterior of \(B\) given \(\Sigma\) is
\[
p(\tvec(B) \mid X, W, \Sigma) \propto \exp \left(-\frac 1 2 \tvec(B - B_n)' V_n^{-1} \tvec(B - B_n) \right), 
\]
where
\begin{align*}
V_n &= \left[ \Sigma^{-1} \otimes X'X + V_0^{-1} \right]^{-1},\\
\tvec(B_n) &= V_n \left[ \left( \Sigma^{-1} \otimes X'X \right) \tvec(\hat{B}) + V_0^{-1} \tvec(B_0) \right],
\end{align*}
and \(\hat{B} =[\hat\delta,\hat\gamma]=(X'X)^{-1}X'W\) is the OLS estimator of \(B\). And the conditional posterior of \(\Sigma\) given \(B\) is
\[
\Sigma \mid X, W, B \sim \IW\left(\nu_n, \,\Sigma_n\right),
\]
where
\begin{align*}
\nu_n &= \nu_0 + n, \\
\Sigma_n &= \Sigma_0 + (W - X B)'(W - X B).
\end{align*}

Below, for notation simplicity, let \(B_0=0\).
\begin{rem}[Posterior: BDML, finite sample] \label{rmk:post-BDML-finite}
  {\normalfont
  We cannot derive the exact posterior distribution of \(\alpha\) in the finite sample case, as there is no closed form for the marginal posterior distribution of \(\Sigma\). However, we can derive the posterior of \(\alpha\) in the limit. 
  }
\end{rem}
\begin{pro}[Marginal posterior distribution of \(\Sigma\)] 
  \label{pro:marginal-post-Sigma}
  Assume that the model and prior are as above. Suppose that Assumptions \ref{assump:dgp}--\ref{assump:rates} hold. We can approximate the marginal posterior of \(\Sigma\) as
  \begin{align*}
  p(\Sigma | X, W) \propto |\Sigma|^{-(\tilde\nu_n + 3)/2} \exp\left(-\frac{1}{2}\tr(\tilde{\Sigma}_n\Sigma^{-1})\right)  \exp\left(\frac{1}{2}\tr(\tilde C_n\Sigma)\right)(1+o_p(1)),
  \end{align*}
  where
  \begin{align*}
  \tilde\nu_n &= \nu_0 + n-p, \\
  \tilde{\Sigma}_n &= \Sigma_0 + (W - X\hat{B})'(W - X\hat{B}), \\
  \tilde C_n &= \begin{bmatrix}
    \tau_{\delta}^2 \hat\delta'(X'X)^{-1}\hat\delta
   - \tau_{\delta}\tr\left((X'X)^{-1}\right) 
   & \tau_{\delta}\tau_{\gamma} \hat\delta'(X'X)^{-1}\hat\gamma \\
  \tau_{\delta}\tau_{\gamma} \hat\gamma'(X'X)^{-1}\hat\delta
   & \tau_{\gamma}^2 \hat\gamma'(X'X)^{-1}\hat\gamma
   - \tau_{\gamma}\tr\left((X'X)^{-1}\right) 
  \end{bmatrix}.
  \end{align*}
\end{pro}

\begin{rem}[Intuition on the posterior of \(\Sigma\)]\label{rmk:post-Sigma}
  {\normalfont
  The \(\exp\left(\frac{1}{2}\tr(\tilde C_n\Sigma)\right)\) term introduces exponential tilting to the Inverse-Wishart distribution. To visualize this effect, we could think of the tilting as ``pulling'' the distribution of \(\Sigma\) in the direction of matrices that have a structure similar to \(\tilde C_n\) (or \(-\tilde C_n^{-1}\)), if \(\tilde C_n\) is positive (or negative) definite.\footnote{If \(\tilde C_n\) is negative definite, then \(X\) follows a multivariate Generalized Inverse Gaussian distribution.}

As \(n\rightarrow\infty\), \(\tau_{\delta},\tau_{\gamma} = o(n)\), \(\tilde C_n = o(n)\), and \(\tilde{\Sigma}_n = O(n)\) so the IW part would dominate the tilting part.
  }
\end{rem}

\begin{pro}[Posterior mean: BDML] \label{pro:post-BDML}
  Assume that the model and prior are as above. Suppose that Assumptions \ref{assump:dgp}--\ref{assump:rates} hold. The posterior mean of \(\alpha|X,W\) under the BDML estimator is given by
  \begin{align*}
  \hat\alpha_\BDML&=\E[\alpha|X,W]\\ &=\alpha^*+\frac {\alpha^*} n\left(\frac{\Sigma_{0,12}+(\Sigma^*C^*\Sigma^*)_{12}}{\Sigma^*_{12}}-\frac{\Sigma_{0,22}+(\Sigma^*C^*\Sigma^*)_{22}}{\Sigma^*_{22}}\right)\\ &\quad+o_p\left(\frac 1 n\right)+o_p\left(\left(\frac p n\right)^2\right)+O_p\left(\frac 1 {\sqrt n}\right), 
  \end{align*}
  where 
  \begin{align*}
    C^* &= \frac 1 n\begin{bmatrix}
        \tau_{\delta}^2 \delta^{*\prime}\Sigma_X^{-1}\delta^*
       - \tau_{\delta}\tr\left(\Sigma_X^{-1}\right) 
       & \tau_{\delta}\tau_{\gamma} \delta^{*\prime}\Sigma_X^{-1}\gamma^* \\
      \tau_{\delta}\tau_{\gamma} \gamma^{*\prime}\Sigma_X^{-1}\delta^*
       & \tau_{\gamma}^2 \gamma^{*\prime}\Sigma_X^{-1}\gamma^*
       - \tau_{\gamma}\tr\left(\Sigma_X^{-1}\right)
        \end{bmatrix},
    \end{align*} 
    and the \(O_p\left(\frac 1 {\sqrt n}\right)\) term is mean zero.
    Therefore, \(\hat\alpha_\BDML\) is consistent. Moreover, if \(p=o(n^{3/4})\), \(\hat\alpha_\BDML\) is \(\sqrt n\)-consistent.
\end{pro}

\begin{rem}[BDML vs.\ na\"ive estimator]\label{rmk:bdml-naive}
  {\normalfont
 Compared with Proposition \ref{pro:post-mean-naive} for the na\"ive estimator, we see that when \(p\) is between the order of \(n^{1/2}\) and \(n^{3/4}\), the BDML estimator is \(\sqrt n\)-consistent, while the na\"ive estimator is not; this is in line with the argument for FDML in Chernozhukov et al.\ (2018, EctJ). In other words, BDML and FDML have bias of order $p^2/n^2$ compared to $p/n$ for Na\"ive.
  }
\end{rem}

\begin{rem}[BDML vs.\ FDML]\label{rmk:bdml-fdml}
{\normalfont
For the frequentist DML based on ridge, for the first stage with \(\tau_{\delta},\tau_{\gamma}=o(n)\), 
\begin{align*}
\tvec(\hat B_\FDML-\hat B) = - V_{n,FDML}V_0^{-1}\tvec(\hat B),
\end{align*}
where 
\begin{align*}
  V_{n,FDML} = \begin{bmatrix}
    [(X'X)^{-1}+\tau_{\delta}I]^{-1} & 0 \\ 0 & [(X'X)^{-1}+\tau_{\gamma}I]^{-1}
  \end{bmatrix}
  = I\otimes (X'X)^{-1}(1+o_p(1)).
  \end{align*}
  Then, 
  \begin{align*}
    \hat B_\FDML-\hat B = - (X'X)^{-1}[\tau_{\delta}\hat\delta,\, \tau_{\gamma}\hat\gamma](1+o_p(1)).
    \end{align*}
 The sample covariance matrix of the residues after the first stage is 
\begin{align*}
(W-X \hat B_\FDML)'(W-X\hat B_\FDML)/n
&=\bar\Sigma+(\hat B_\FDML-\hat B)' (X'X) (\hat B_\FDML-\hat B)\\
&=\bar\Sigma+\tilde C_{n,2}/n(1+o_p(1)),
\end{align*}
where 
\begin{align*}
\bar\Sigma &= (W - X\hat{B})'(W - X\hat{B}),\\
  \tilde C_{n,2}&=\begin{bmatrix}
    \tau_{\delta}^2\hat\delta'(X'X)^{-1}\hat\delta & \tau_{\delta}\tau_{\gamma}\hat\delta'(X'X)^{-1}\hat\gamma\\
    \tau_{\delta}\tau_{\gamma}\hat\gamma'(X'X)^{-1}\hat\delta & \tau_{\gamma}^2\hat\gamma'(X'X)^{-1}\hat\gamma
    \end{bmatrix},
\end{align*}
Then, 
\begin{align*}
  \hat\alpha_\FDML&=\frac{\bar\Sigma_{12}+\tilde C_{n,2,12}/n(1+o_p(1))}
  {\bar\Sigma_{22}+\tilde C_{n,2,22}/n(1+o_p(1))}\\ 
  &=\alpha^*+\frac {\alpha^*} n\left(\frac{(C^*_{2})_{12}}{\Sigma^*_{12}}-\frac{(C^*_{2})_{22}}{\Sigma^*_{22}}\right)+o_p\left(\frac 1 n\right)+o_p\left(\left(\frac p n\right)^2\right)+O_p\left(\frac 1 {\sqrt n}\right),
\end{align*}
where \(\tilde C_{n,2} = C^*_{2}(1+o_p(1))\) with
\begin{align*}
  C^*_{2} &= \frac 1 n\begin{bmatrix}
    \tau_{\delta}^2 \delta^{*\prime}\Sigma_X^{-1}\delta^*
   & \tau_{\delta}\tau_{\gamma} \delta^{*\prime}\Sigma_X^{-1}\gamma^* \\
  \tau_{\delta}\tau_{\gamma} \gamma^{*\prime}\Sigma_X^{-1}\delta^*
   & \tau_{\gamma}^2 \gamma^{*\prime}\Sigma_X^{-1}\gamma^*
    \end{bmatrix},
\end{align*}
and the \(O_p\left(\frac 1 {\sqrt n}\right)\) term is mean zero. 

Comparing with the BDML case in Proposition \ref{pro:post-BDML}, the formula takes a similar form. Note that \(C^*_2=O(p^2/n)\), so the FDML estimator is also consistent, and \(\sqrt n\)-consistent if \(p=o(n^{3/4})\).
The higher order differences being
\begin{itemize}
  \item BDML has an extra \(\Sigma_0/n\) term from the prior of the error covariance. 
  \item FDML incorporates \(C^*_{2}\), which is the second component of the tilting term in the BDML case in \eqref{eq:tilde-C2}. BDML has an extra first component of the tilting term in \eqref{eq:tr-sigma1}, coming from the first order approximation of \(|V_n|^{1/2}\).
\item In addition, the last term in BDML is based on \(\Sigma^*C^*\Sigma^*\). While the variance terms in \(\Sigma^*\) can be absorbed into \(\tau_{\delta}\) and \(\tau_{\gamma}\), the correlation is not captured in FDML.
\end{itemize}
}
\end{rem}

For the bias-variance trade-off, note that the variance of the posterior mean is not the same as the posterior variance. As we only have posterior variance for the BDML estimator, below we focus on deriving and comparing the variance of the posterior mean.

\begin{pro}[Variances of \(\hat\alpha\)]
  \label{pro:var-alpha}
  Assume that the model and prior are as above. Suppose that Assumptions \ref{assump:dgp}--\ref{assump:rates} hold. The asymptotic variance of \(\hat\alpha_\naive\), \(\hat\alpha_\BDML\) and \(\hat\alpha_\FDML\) are the same, and are given by 
  \[
  \frac{\Sigma^*_{11}\Sigma^*_{22}-(\Sigma^*_{12})^2}{(\Sigma^*_{22})^2}.
  \]
  \end{pro}

\begin{rem}[Comparison of variances]\label{rmk:var-alpha}
  {\normalfont
  The variances of \(\hat\alpha_\BDML\), \(\hat\alpha_\FDML\), and \(\hat\alpha_\naive\) are the same. This is because the variances of the error terms are the same, which is the dominating term in the variance of \(\hat\alpha\). Therefore, to minimize MSE, we would focus on the bias of the estimators, where the BDML and FDML estimators are \(\sqrt n\)-consistent in more general cases.
  }
\end{rem}

\subsection{Bernstein–von Mises theorem}  

Under the assumption that $\sigma_{\varepsilon}^{2*}$ and $\sigma_{V}^{2*}$ are known, we can resort to Walker (2024) by verifying his assumptions in our high dimensional regression setup.

\begin{rem}[Advantages of BDML: BvM]\label{rmk:advantages-bvm}
  {\normalfont
  The procedure is semiparametrically efficient in that it is adaptive, incurring no cost from not knowing the nuisance parameters. It avoids the need for prior invariance and does not require strict smoothness conditions on the propensity score.

In addition, the method is robust due to its use of a Gaussian profile likelihood; it only requires the mean restriction—similar to ordinary least squares. This property ensures that Bayesian credible intervals remain efficient even under model misspecification.}
\end{rem} 

Let $P^*$ denote the probability measure based on the true DGP, and $\|\cdot\|_{TV}$ denote the total variation distance.

\begin{pro}[Bernstein--von Mises theorem: known $\sigma_{\varepsilon}^{2*}$ and $\sigma_{V}^{2*}$]
  \label{pro:bvm-known-sigma}
  Assume that the model and prior are as above, except that $\sigma_{\varepsilon}^{2*}$ and $\sigma_{V}^{2*}$ are known, the prior for $\alpha$ is continuous and positive over a neighborhood of $\alpha^*$, and  $\alpha^*$ is in the interior of its support. Suppose that Assumptions \ref{assump:dgp}--\ref{assump:rates} hold. Further suppose that $X_i$ is sub-Gaussian, and the eigenvalues of $\Sigma_X$ is uniformly bounded by $[\underline\lambda_X,\bar\lambda_X]$ with $0<\underline\lambda_X\le \bar\lambda_X<\infty$, $p=o\left(\sqrt{n}\right)$, and $\log n = o(p^{1/2})$. Then, as $n\rightarrow\infty$, 
  \[\left\|P\left(\alpha\in\cdot|X,W\right)-\normal\left(\alpha^*+\frac{\tilde\Delta_n^*}{\sqrt n},\,\frac 1 n I_n(\gamma^*)^{-1}\right)\right\|_{TV}\pstarto 0,\]
  where $\tilde I_n(\gamma)=\avgi\frac{(D_i-X_i'\gamma)^2}{\sigma_{\varepsilon}^{2*}}$ and $\tilde\Delta_n^*=\tilde I_n(\gamma^*)^{-1}\sqrti\frac{[Y_i-\alpha^*(D_i-X_i'\gamma^*)-X_i'\delta^*] (D_i-X_i'\gamma^*)}{\sigma_{\varepsilon}^{2*}}$.
\end{pro}

\begin{rem}[BvM: empirical $L_2$-norm]\label{rmk:bvm-norm}
  {\normalfont
  Here I consider the empirical $L_2$-norm in the definition of the shrinking neighborhood of $m_0$ to accommodate our linear setup with unbounded $X_i$. Specifically, for generic functions $f$ and $g$, the empirical $L_2$ inner product is defined as $\langle f,g\rangle_{n,2}=\avgi f(X_i)g(X_i)$, and the empirical $L_2$-norm is $\|f\|_{n}=\sqrt{\langle f,f\rangle_{n,2}}$. Then, the neighborhood of $m_0$ based on the empirical $L_2$-norm is defined as \[B_{\mathcal M,n,2}(m^*,\nu)=\{m\in\mathcal M:\,\|m-m^*\|_{n,2}\le\nu\},\] which is different from $B_{\mathcal M,\infty}(m^*,\nu)$ in Walker (2024) for the sup-norm. All proofs in Walker (2024) can be adapted to the empirical $L_2$-norm: see also the last sentence on page 8 in Walker (2024).
  }
\end{rem}

\begin{rem}[BvM: alternative conditions]\label{rmk:bvm-alt-cond}
  {\normalfont
  One stronger alternative condition is that the eigenvalues of $X'X/n$ are bounded away from zero and infinity, so we don't need matrix Bernstein and $\mathcal K_n=\mathbb R^p$ below.

  Another alternative condition is that $X_i$ is bounded, i.e., for some constant \(M_X>0\),
  \(
  \max_{1\le j\le p} |X_{i,j}| \le M_X.
  \) This is necessary for the sup-norm case, but incurs another $\sqrt p$ term in $\nu_n$, and only allows for $p=o(n^{1/4})$.
  }
\end{rem}

For a generic matrix $S$, define the Orlicz \(\psi_1\) norm as
\[
\|S\|_{\psi_1} = \inf\left\{s>0:\,\E\exp\left(\frac{\|S\|}{s}\right)\le 2\right\},
\]
where $\|\cdot\|$ is the spectral norm. Note that $\|S\|_{\psi_1}$ is finite if and only if $S$ is sub-exponential. Then, the matrix Bernstein inequality is as follows.

\begin{lem}[Matrix Bernstein]\label{lem:matrix-bernstein}
  Let \(S_1,\ldots,S_n\) be i.i.d., mean-zero, symmetric random matrices in \(\mathbb{R}^{p\times p}\), satisfying that
  \(
  \|S_i\|_{\psi_1} \le M_S
  \) for all $i$.
  Then, for any \(\epsilon \ge 0\),
  \[
  P\left\{ \left\|\sumi S_i\right\| \ge \epsilon \right\}
  \le 2p \cdot \exp\left( - c\cdot\min\left( \frac{\epsilon^2}{nM_S^2},\,\frac{\epsilon}{M_S} \right)\right),
  \]
  where \(c>0\) is an absolute constant.
\end{lem}

\begin{cor}[Eigenvalue bounds]\label{cor:eigen-bounds}
  Suppose that $X_i$ is sub-gaussian and i.i.d., and the eigenvalues of $\Sigma_X$ is uniformly bounded by $[\underline\lambda_X,\bar\lambda_X]$ with $0<\underline\lambda_X\le \bar\lambda_X<\infty$. Define event \[\mathcal K_n =\left\{X\in\mathbb R^{n\times p}:\,\lambda_{\min}\left(\frac{X'X} n\right)\ge \frac{\underline\lambda_X}2\text{ and }\lambda_{\max}\left(\frac{X'X} n\right)\le 2\bar\lambda_X\right\}.\]
  Then, $P(\mathcal K_n)\pto 1$ and $P(\mathcal K_n^c)=O(p^{-1})$, 
  as $n\to\infty$.
\end{cor}

%!TEX root = ./main.tex

\section{Simulation Study}
\label{sec:sim}
We now compare our proposed BDML to a number of alternatives from the literature in a simulation experiment adapted from  \cite{linero_nonparametric_2023}. 
In each replication we generate 
\begin{gather}
\begin{aligned}
\left\{X_i\right\}_{i=1}^n &\sim \text{iid Normal}_p(0, \mathbbm{I}_p)\\
\left\{(\varepsilon_i, V_i)'\right\}_{i=1}^n\mid X &\sim \text{iid Normal}\left(\begin{bmatrix} 0 \\ 0\end{bmatrix}, \begin{bmatrix} \sigma_\varepsilon^2 & 0 \\ 0 & 1\end{bmatrix} \right) \\
\beta\mid(X, \varepsilon, V) &\sim \text{Normal}_p\left(\mu_\beta, \sigma_\beta^2\mathbbm{I}_p\right).
\end{aligned}
\label{eq:SimDGP}
\end{gather}
We then construct $(D_i, Y_i)$ according to \eqref{eq:linearY}--\eqref{eq:linearD} above.
This design generates new control regressors $X$ and a new parameter vector $\beta$ in each replication but holds $\gamma$ and $\alpha$ constant across replications.
Following the ``fixed'' design from \cite{linero_nonparametric_2023}, we vary $\sigma_\varepsilon$ over a grid of values from $1$ to $4$ and set
\begin{equation}
  \alpha = 2, \quad
  \gamma = \iota_p / \sqrt{p}, \quad
  \mu_\beta = -\gamma/2, \quad
  \sigma_\beta^2 = 1/p, \quad
  n = 200, \quad
  p= 100
  \label{eq:SimDesign1}
\end{equation}
where $\iota_p$ denotes a $p$-vector of ones.

We consider three groups of estimators.
The first group includes two alternative versions of the procedure from \autoref{alg:linearBDML}, corresponding to alternative prior structures for $(\gamma, \delta)$.
Both of our BDML approaches place a flat prior on $\alpha$ and construct a weakly informative prior for $\Sigma$ by placing an LKJ$(4)$ prior on the correlation matrix for $(U,V)$ and independent Cauchy$(0, 2.5)$ priors on both $\sigma_U$ and $\sigma_V$. 
The first of our BDML approaches, \textbf{BDML-Basic}, places independent $\text{Normal}(0, 5^2)$ on the elements of $\delta$ and $\gamma$. 
The second, \textbf{BDML-Hier}, allows different standard deviations in the normal shrinkage priors for $\delta$ and $\gamma$.
We achieve this via a hierarchical prior that places independent Inverse-Gamma$(2,2)$ hyper-priors on $\sigma_\delta^2$ and $\sigma_\gamma^2$.\footnote{This is equivalent to placing independent Student-t$(4)$ distributions on each of the coefficients $\delta_j$ and $\gamma_j$.}
%Our theoretical derivations suggest that the hierarchical approach should give better performance and, as we see below, this is borne out in our simulations.
We sample from the posteriors for both BDML approaches via Hamiltonian Monte Carlo (HMC), using the no-U-turn (NUTS) sampler as implemented in the Stan probabilistic programming language.

The second group of estimators we consider contains Bayesian competitors to our BDML approach.
This group includes three estimators.
The first two, \textbf{HCPH} and \textbf{Linero}, correspond to the proposals from \cite{hahn_regularization_2018} and \cite{linero_nonparametric_2023} detailed in Algorithms \ref{alg:hahn}--\ref{alg:linero} above and rely on a first-stage regression of $D$ on $X$.
The third, \textbf{Na\"ive}, estimates a single-equation linear regression model including \eqref{eq:linearY} only and is a fully-Bayesian variant of the ridge regression approach described in \autoref{sec:RIC} above.
We carry out posterior sampling for these three approaches exactly as in \cite{linero_nonparametric_2023}.
In particular, we rely on the Gibbs sampler implemented in the \texttt{BLR()} function from the R package \texttt{BLR} for Bayesian linear regression.
Again following \cite{linero_nonparametric_2023} we place a flat prior on $\alpha$ across all three approaches and on $\kappa$ from \autoref{alg:linero} and a Jeffreys prior on all error variances.
We likewise place default ``ridge priors'' on the coefficients on $X$ across all regressions in both stages: independent Normal$(0, \sigma^2)$ priors on each coefficient and a Jeffreys prior on $\sigma^2$.\footnote{The resulting prior for the coefficients on $X$ is improper and sharply peaked at zero.}

%\textcolor{red}{Explain priors used for all of these!}
%\begin{itemize}
%  \item First-Step for Hahn/Linero: \verb|fitted_ps <- BLR(y = A, XR = X)|
%  \item Hahn: \verb|BLR(y = Y, XF = cbind(A-fitted_ps$yHat), XR = X)|
%  \item Linero: \verb|BLR(y = Y, XF = cbind(A, fitted_ps$yHat), XR = X)|
%  \item Naive: \verb|BLR(y = Y, XF = as.matrix(A), XR = X)| 
%\end{itemize}

The third and final group of estimators contains frequentist DML competitors to our BDML approach. 
Both of the approaches in this group estimate $\alpha$ via \eqref{eq:FDML} where the ``first-step'' estimators $\widehat{\gamma}$ and $\widehat{\delta}$ are constructed via ridge regression.\footnote{We use ridge rather than LASSO for these preliminary DML regressions because recent work by \cite{kolesar_fragility_2025} and \cite{shen_can_2024} suggests that the ``double-LASSO'' approach \citep{belloni_inference_2014} perform poorly due to its fragility in variable selection.}
We carry out ridge regression by computing the posterior mean from draws produced using \texttt{BLR()} with the default parameter values, as explained above.
We consider two versions of this FDML approach: ``full-sample'' \textbf{FDML-Full} and ``split-sample'' \textbf{FDML-Split}.
\textbf{FDML-Full} uses the full dataset for both estimation stages whereas \textbf{FDML-Split} randomly splits the sample partitions the sample into two halves, the first of which is used to estimate $(\widehat{\delta}_\text{ML}, \widehat{\gamma}_\text{ML})$ and the second of which is used to residualize $(Y,D)$ and estimate $\alpha$.

For each of the seven estimators of $\alpha$ described above, we compute the root mean-squared error (RMSE), along with the frequentist coverage probability and average width of a corresponding (nominal) 95\% confidence interval.
For the five Bayesian procedures, interval estimates are constructed as equal-tailed 95\% posterior credible intervals.
For the two frequentist DML procedures, confidence intervals are constructed based on the usual OLS interval for the regression of the residualized $Y$ on the residualized $D$, following the approach suggested by \cite{belloni_inference_2014} and often used in empirical research.
Results based on 200 simulation replications from the data-generating process given in \eqref{eq:SimDGP}--\eqref{eq:SimDesign1} appear in \autoref{tab:sim} and \autoref{fig:sim}.
Coverage probabilities are very poor for all methods except BDML-Basic, BDML-Hier, and Linero.
All three of these approaches produce coverage probabilities close to the nominal 95\% level, but BDML-Hier is closest.
From panel \ref{fig:sim-zoom-in} of \autoref{fig:sim}, we also see that BDML-Hier produces the shortest average intervals among the three methods with approximately correct coverage.
This holds across all values of $\sigma_\varepsilon$ in our simulation.
Turning our attention to RMSE, the three best-performing methods remain BDML-basic, BDML-Hier, and Linero.
Once again from panel \ref{fig:sim-zoom-in} of \autoref{fig:sim}, we see that BDML-Hier performs best.
The fact that BDML out-peforms Linero in this setting is particularly encouraging given that this simulation design was taken directly from \cite{linero_nonparametric_2023} and favors \autoref{alg:linero}.

\begin{table}[ht]
\centering
\begin{tabular}{lccccccc}
  \hline
  \hline
  Method &  $\sigma_\varepsilon$ & $n$ & $p$ & Coverage & RMSE & Avg.\ Width \\ 
  \hline
  BDML-Hier & 1& 200& 100& 0.94 & 0.09 & 0.36 \\ 
  BDML-Hier & 2& 200& 100& 0.94 & 0.18 & 0.66 \\ 
  BDML-Hier & 4& 200& 100& 0.94 & 0.35 & 1.28 \\ 
  BDML-Basic & 1& 200& 100& 0.93 & 0.11 & 0.41 \\ 
  BDML-Basic & 2& 200& 100& 0.91 & 0.22 & 0.80 \\ 
  BDML-Basic & 4& 200& 100& 0.92 & 0.46 & 1.54 \\ 
  \hline
  Linero & 1& 200& 100& 0.93 & 0.10 & 0.38 \\ 
  Linero & 2& 200& 100& 0.93 & 0.20 & 0.76 \\ 
  Linero & 4& 200& 100& 0.93 & 0.39 & 1.49 \\ 
  HCPH & 1 & 200& 100& 0.65 & 0.19 & 0.38 \\ 
  HCPH & 2 & 200& 100& 0.56 & 0.39 & 0.73 \\ 
  HCPH & 4 & 200& 100& 0.68 & 0.63 & 1.37 \\ 
  Na\"ive & 1& 200& 100& 0.49 & 0.17 & 0.29 \\ 
  Na\"ive & 2& 200& 100& 0.56 & 0.26 & 0.47 \\ 
  Na\"ive & 4& 200& 100& 0.73 & 0.34 & 0.83 \\ 
  \hline
  FDML-Full & 1& 200& 100& 0.82 & 0.13 & 0.31 \\ 
  FDML-Full & 2& 200& 100& 0.69 & 0.29 & 0.61 \\ 
  FDML-Full & 4& 200& 100& 0.71 & 0.56 & 1.27 \\ 
  FDML-Split & 1& 200& 100& 0.56 & 0.22 & 0.42 \\ 
  FDML-Split & 2& 200& 100& 0.79 & 0.28 & 0.70 \\ 
  FDML-Split & 4& 200& 100& 0.88 & 0.41 & 1.29 \\ 
   \hline
\end{tabular}
\caption{This table presents results for the estimators described in \autoref{sec:sim} and simulation from \eqref{eq:SimDGP}--\eqref{eq:SimDesign1} over a grid of values for $\sigma_\varepsilon$, the error variance in the structural equation given in \eqref{eq:linearY}. The column RMSE gives the root mean-squared error of the estimator for $\alpha$ under each method, while Coverage and Avg.\ Width give the coverage probability and average width of corresponding (nominal) 95\% confidence intervals. Results are based on 200 simulation replications.}
\label{tab:sim}
\end{table}

\begin{figure}
    \centering
    % First panel
    \begin{subfigure}{\textwidth}
        \centering
        \includegraphics[width=0.999\textwidth]{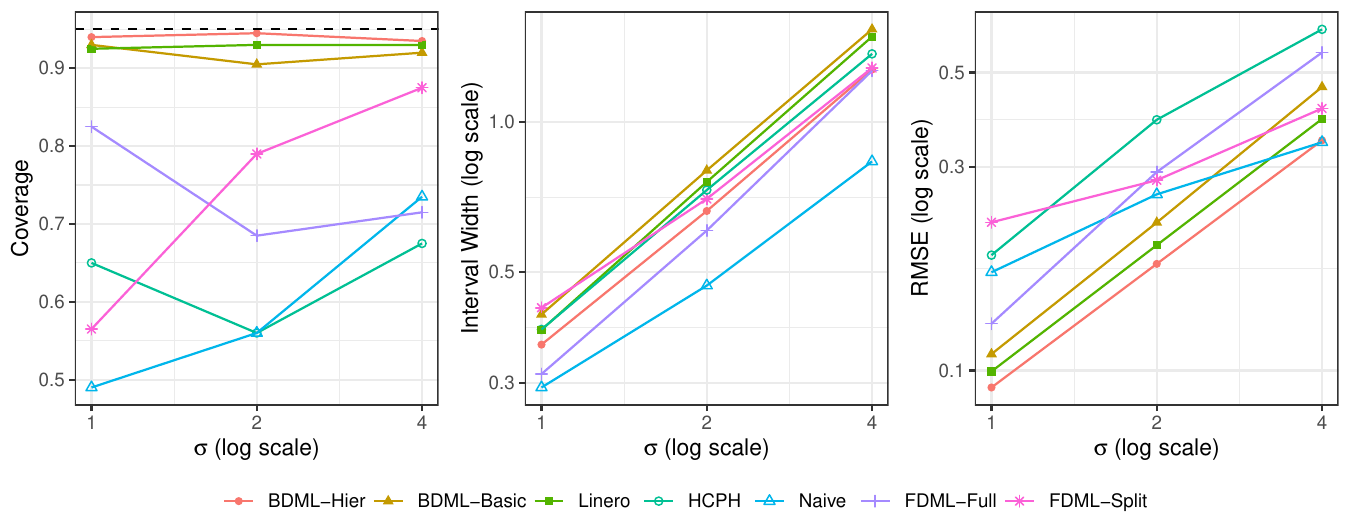}
        \caption{}
        \label{fig:sim-zoom-out}
    \end{subfigure}
    
    \vspace{1em}
    
    % Second panel
    \begin{subfigure}{\textwidth}
        \centering
        \includegraphics[width=0.999\textwidth]{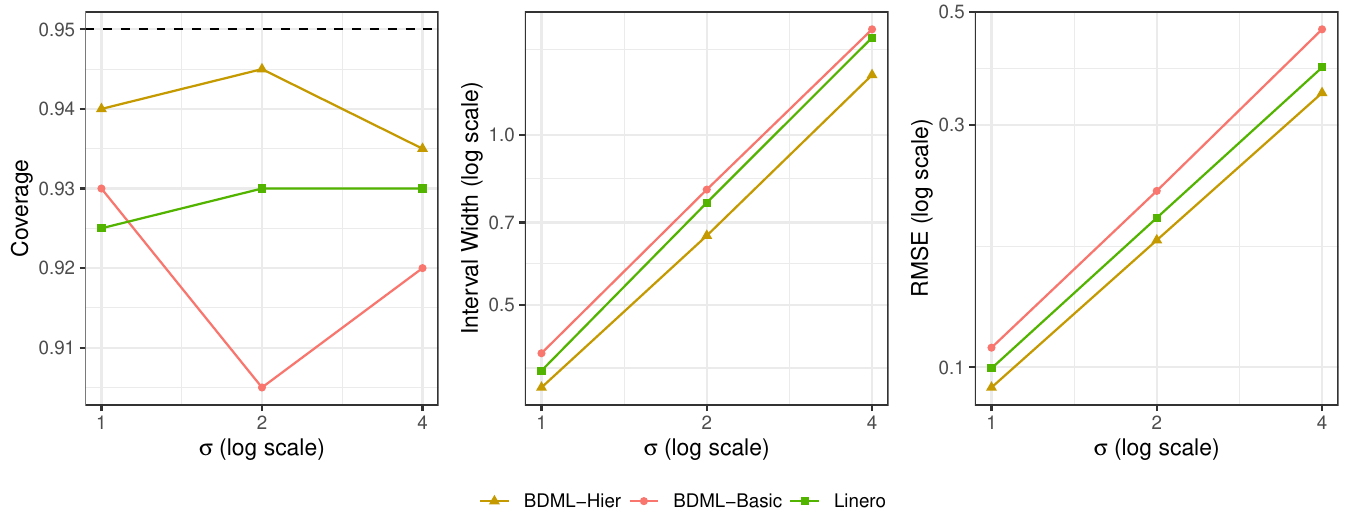}
        \caption{}
        \label{fig:sim-zoom-in}
    \end{subfigure}
    
    \caption{This figure presents results for the estimators described in \autoref{sec:sim} and simulation from \eqref{eq:SimDGP}--\eqref{eq:SimDesign1} over a grid of values for $\sigma_\varepsilon$, the error variance in the structural equation given in \eqref{eq:linearY}. The panel labled RMSE gives the root mean-squared error of the estimator for $\alpha$ under each method, while those labled Coverage and Avg.\ Width give the coverage probability and average width of corresponding (nominal) 95\% confidence intervals. Panel \ref{fig:sim-zoom-out} plots results on a scale such that all seven approaches are visible; panel \ref{fig:sim-zoom-in} zooms in to more clearly show the differences in performance betwen the three best estimators. Results are based on 200 simulation replications.}
    \label{fig:sim}
\end{figure}

%\todo[inline]{Notes / prelim material below here} 
%\paragraph{Linero's Simulation Code}
%\begin{itemize}
%  \item Use the ``fixed'' DGP of Linero: $\gamma = 2$ and $w = -0.5$ ($\omega$) 
%  \item \verb|beta <- rnorm(P, 0, 1/sqrt(P))|
%  \item \verb|X <- matrix(rnorm(n = N * P), nrow = N)|
%  \item \verb|theta <- rep(1/sqrt(P), P)|
%  \item \verb|A_hat <- as.numeric(X %*% theta)|
%  \item \verb|A <- A_hat + rnorm(N)|
%  \item \verb|mu_hat <- as.numeric(X %*% beta)|
%  \item \verb|Y <- gamma * A + w * A_hat + mu_hat + sigma * rnorm(N)|
%\end{itemize}
%
%\paragraph{Translating to Equations}
%\begin{align*}
%  \gamma &= 2\\
%  \omega &= -0.5\\
%  \theta_1, \dots, \theta_p &= 1/\sqrt{p}\\ 
%  \beta_1, \dots, \beta_p &\sim \text{iid Normal}(0, \sigma_\beta = 1/\sqrt{p})\\
%  X_1, \dots, X_n &\sim \text{iid Normal}_p(\textbf{0}, \textbf{I}_p)\\
%  \hat{A}_i &= X_i'\theta\\
%  A_i &= \hat{A}_i + [\text{Normal}(0,1)]_i\\
%  \hat{\mu}_i &= X_i'\beta\\
%  Y_i &= \gamma A_i + \omega \hat{A}_i + \hat{\mu}_i + \sigma \times [\text{Normal}(0,1)]_i
%\end{align*}
%
%\paragraph{Simplify the Notation}
%\begin{align*}
%  X_i &\sim \text{iid Normal}_p(\textbf{0}, \textbf{I}_p)\\
%  V_i &\sim \text{iid Normal}(0,1)\\
%  U_i &\sim \text{iid Normal}(0,1)\\
%  A_i &= X_i'\theta + V_i\\
%  Y_i &= \gamma A_i + \omega X_i'\theta + X_i'\beta + \sigma U_i\\
%  &= \gamma A_i + X_i'(\omega \theta + \beta) + \sigma U_i\\
%  &= \gamma A_i + X_i'\delta + \varepsilon_i\\
%  \delta_j &\equiv (\omega \theta_j + \beta_j) \sim \text{iid Normal}(\omega \theta_j, \sigma_\beta)
%\end{align*}

\clearpage

% References
%\small
\singlespacing
\bibliographystyle{elsarticle-harv}
\bibliography{bayesian-DML}

%Proofs in the Appendix
\appendix
\numberwithin{equation}{section}
\numberwithin{table}{section}
\numberwithin{figure}{section}

%Appendices
%!TEX root = ./main.tex
\clearpage
\section{Proofs}

\label{append:proofs}

The following lemma is a preliminary step in the proof of \autoref{pro:ridgeMSE}.

\begin{lem}
\label{lem:ridgeMSE}
Define the shorthand 
\[
  M = \begin{pmatrix} M_{11} & M_{12}\\
  M_{21} & M_{22}\end{pmatrix} \equiv \begin{pmatrix} D'D & D'X \\
    X'D & X'X + \lambda \mathbbm{I}_p
  \end{pmatrix}, \quad
  M^{-1} \equiv \begin{pmatrix}
    M^{11} & M^{12} \\
    M^{21} & M^{22}
  \end{pmatrix}
\]
and let $\widehat{\rho}' \equiv (D'D)^{-1}D'X$, $\widehat{\xi} \equiv \left[ \mathbbm{I}_p - D(D'D)^{-1}D' \right]X$ and $R = \widehat{\xi}'\widehat{\xi}$.
Then we have 
\begin{align}
  \label{eq:M12}
  M^{12} &= -\widehat{\rho}'(R + \lambda I)^{-1}\\
  M^{11} &= (D'D)^{-1} - M^{12} \widehat{\rho}\\
  \begin{pmatrix}
    M^{11} & M^{12}
  \end{pmatrix}
  \begin{pmatrix}
    D' D & D'X\\ X'D & X'X
  \end{pmatrix} &= 
  \begin{bmatrix}
    1 & (\widehat{\rho}'+ M^{12} R )
  \end{bmatrix}.
  \label{eq:M11M12}
\end{align}
\end{lem}

\begin{proof}[Proof of \autoref{lem:ridgeMSE}]
By the partitioned matrix inverse formula,
\[
  M^{-1} = \begin{pmatrix}
    M^{11} & M^{12} \\
    M^{21} & M^{22}
  \end{pmatrix} 
  = \begin{pmatrix}
    M_{11}^{-1} + M_{11}^{-1}M_{12}S^{-1}M_{21}M_{11}^{-1} & -M_{11}^{-1}M_{12}S^{-1} \\
    -S^{-1}M_{21}M_{11}^{-1} & S^{-1}
\end{pmatrix}
\]
where we define 
\[
  S \equiv M_{22} - M_{21}M_{11}^{-1}M_{12}.
\]
The preceding expression for $M^{-1}$ holds so long as $M_{11}$ and $S$ are invertible.
We see that $M_{11} \equiv D'D$ is invertible whenever there is variation in the treatment variable while $M_{22} \equiv X'X + \lambda \mathbbm{I}_p$ is invertible when $\lambda > 0$ although possibly not when $\lambda = 0$, since $p$ could be larger than the sample size.
Now, defining $P_D \equiv D(D'D)^{-1}D'$, we have
\begin{align*}
  S &= M_{22} - M_{21} M_{11}^{-1}M_{12} = (X'X + \lambda \mathbbm{I}_p) - X'D (D'D)^{-1}D'X\\
  &= X' \left(\mathbbm{I}_p - P_D\right)X + \lambda \mathbbm{I}_p\\
  &= \left[\left(\mathbbm{I}_p - P_D\right)X\right]'\left[\left(\mathbbm{I}_p - P_D\right)X\right]+ \lambda \mathbbm{I}_p\\
  &= \widehat{\xi}' \widehat{\xi} + \lambda \mathbbm{I}_p\\
  &= R + \lambda \mathbbm{I}_p
\end{align*}
since $P_D$ is a projection matrix and $(\mathbbm{I}_p - P_D)$ is its orthogonal complement.
Substituting this simplified expression for $S$ into the expression for $M^{12}$ from above gives 
\[
  M^{12} = -M_{11}^{-1}M_{12} S^{-1} = -(D'D)^{-1} D'X (R + \lambda \mathbbm{I}_p)^{-1} = -\widehat{\rho}'(R + \lambda \mathbbm{I}_p)^{-1}.
\]
Substituting the expression for $M^{12}$ from the partitioned matrix inverse formula into the expression for $M^{11}$
\[
  M^{11} = M_{11}^{-1} + M_{11}^{-1}M_{12} S^{-1} M_{21} M_{11}^{-1} = M_{11}^{-1} - M^{12} M_{21} M_{11}^{-1}.
\]
It follows that $M^{11}$ can be written as
\begin{align*}
  M^{11} &=  M_{11}^{-1} - M^{12} M_{21} M_{11}^{-1} \\
  &= (D'D)^{-1} + M^{12}X'D (D'D)^{-1}\\
  &= (D'D)^{-1} + M^{12}\left[(D'D)^{-1}D'X\right]'\\
&= (D'D)^{-1} - M^{12}\widehat{\rho}.
\end{align*}
Now, using this more compact expression for $M^{11}$, we obtain
\begin{align*}
  \begin{bmatrix}
    M^{11}& M^{12}
  \end{bmatrix} 
  \begin{bmatrix}
    D'D\\ 
    X' D 
  \end{bmatrix} 
  &=  \begin{bmatrix}
    \left\{(D'D)^{-1} - M^{12}\widehat{\rho}\right\} & M^{12}
  \end{bmatrix} 
  \begin{bmatrix}
    D'D\\ 
    X' D 
  \end{bmatrix} \\
  &= \left[ (D'D)^{-1}- M^{12}\widehat{\rho} \right]D'D + M^{12} X'D \\
  &= 1 - M^{12}\widehat{\rho}D'D + M^{12} X'D \\
  &= 1 - M^{12}\left( X'D - \widehat{\rho} D'D \right)\\ 
  &= 1 - M^{12}\left[ X'D - \left\{ \left( D'D \right)^{-1}D'X \right\}' D'D \right]\\ 
  &= 1 - M^{12}\left[ X'D - X'D\left( D'D \right)^{-1} D'D \right]\\ 
  &= 1 - M^{12}\left( X'D - X'D \right)\\ 
  &= 1.
\end{align*}
Similarly, we have
\begin{align*}
  \begin{bmatrix}
    M^{11} & M^{12}
  \end{bmatrix} 
  \begin{bmatrix}
     D'X\\ 
     X'X
  \end{bmatrix} 
  &=  \begin{bmatrix}
    \left\{(D'D)^{-1} - M^{12}\widehat{\rho}\right\} & M^{12}
  \end{bmatrix} 
  \begin{bmatrix}
     D'X\\ 
      X'X
  \end{bmatrix} \\
&= \left[ (D'D^{-1}) - M^{12}\widehat{\rho} \right] D'X + M^{12} X'X\\
  &= \widehat{\rho}' - M^{12}\widehat{\rho} D'X + M^{12}X'X \\
  &= \widehat{\rho}' + M^{12}\left(X'X - \widehat{\rho}D'X \right) \\
  &= \widehat{\rho}' + M^{12}\left[X'X - X'D\left( D'D \right)^{-1}D'X \right] \\
  &= \widehat{\rho}' + M^{12}\left(X'X - X'P_DX \right) \\
  &= \widehat{\rho}' -\widehat{\rho}' \left( R + \lambda \mathbbm{I}_p \right)^{-1}R\\ 
  &= \widehat{\rho}' \left[ \mathbbm{I}_p -  \left(R + \lambda \mathbbm{I}_p \right)^{-1}R\right].
\end{align*}
The result follows.
\end{proof}

\begin{proof}[Proof of \autoref{pro:ridgeMSE}]
  Using the shorthand defined in the statement of \autoref{lem:ridgeMSE}, 
\[
  \widehat{\alpha}_{\lambda} = 
  \begin{bmatrix}
    M^{11} & M^{12}
  \end{bmatrix} 
  \begin{bmatrix}
    D' \\ X'
  \end{bmatrix} 
  Y 
  = 
  \begin{bmatrix}
    M^{11} & M^{12}
  \end{bmatrix} 
  \begin{bmatrix}
    D' \\ X'
  \end{bmatrix} 
  \left(\begin{bmatrix}
    D & X
  \end{bmatrix} 
  \begin{bmatrix}
    \alpha \\ \beta
\end{bmatrix} + \epsilon \right).
\]
Since $\mathbbm{E}[\varepsilon|X, D]=0$ it follows that 
\[
  \text{Bias}(\widehat{\alpha}_\lambda|X, D) = 
  \begin{bmatrix}
    M^{11} & M^{12}
  \end{bmatrix} 
  \begin{bmatrix}
    D'D & D'X\\ 
    X' D &  X'X
  \end{bmatrix} 
  \begin{bmatrix}
    \alpha \\ \beta
  \end{bmatrix} - \alpha.
\]
Thus, substituting \eqref{eq:M12} and \eqref{eq:M11M12},
\[
\text{Bias}(\widehat{\alpha}_\lambda|X, D) =
  \begin{bmatrix}
    1 & (\widehat{\rho}'+ M^{12} R )
  \end{bmatrix}  \begin{bmatrix}
    \alpha \\ \beta
  \end{bmatrix} - \alpha = \widehat{\rho}' \left[ \mathbbm{I}_p - (R + \lambda \mathbbm{I}_p)^{-1}R \right] \beta.
\]
For the variance calculation, since $\text{Var}(\varepsilon|X, D)=\sigma^2 \mathbbm{I}_p$, we obtain 
\begin{align*}
  \text{Var}\left( \widehat{\alpha}_\lambda|X, D \right) &=
  \text{Var}\left\{ \left.\begin{pmatrix}
      M^{11} & M^{12}
    \end{pmatrix} \begin{pmatrix}
      D' \\ X'
    \end{pmatrix}\left[ \begin{pmatrix}
        D & X
    \end{pmatrix}\begin{pmatrix}
  \alpha \\ \beta
  \end{pmatrix} + \epsilon\right]\right| X, D\right\}\\
  &= \sigma^2 
  \begin{bmatrix}
    M^{11} & M^{12}
  \end{bmatrix}
  \begin{bmatrix}
    D' D & D'X\\ X'D & X'X
  \end{bmatrix}
  \begin{bmatrix}
    M^{11} & M^{12}
  \end{bmatrix}'\\
  &= \sigma^2 \begin{bmatrix}
    1 & (\widehat{\rho}'+ M^{12} \widehat{\xi}'\widehat{\xi})
  \end{bmatrix}\begin{bmatrix}
    (M^{11})' \\
    (M^{12})'
  \end{bmatrix}\\
  &= \sigma^2 \left[ \left\{ (D'D)^{-1} - M^{12}\widehat{\rho} \right\}' + (\widehat{\rho}' + M^{12}R) (M^{12})' \right]\\
  &= \sigma^2\left[ (D'D)^{-1} + M^{12} R(M^{12})' \right]\\
  &= \sigma^2 \left[ (D'D)^{-1} + \widehat{\rho}'(R + \lambda \mathbbm{I}_p)^{-1} R(R + \lambda \mathbbm{I}_p)^{-1} \widehat{\rho}\right] 
\end{align*}
by substituting \autoref{lem:ridgeMSE}, since $R \equiv \widehat{\xi}'\widehat{\xi}$ is symmetric.
\end{proof}

\begin{pro}[Signal-to-noise ratio and $R^2$ of high-dimensional regression]  
  \label{pro:rsquared}
  Suppose that Assumptions \ref{assump:dgp}--\ref{assump:rates} hold. As $p\to\infty$, the signal-to-noise ratios of the $Y$ and $D$ regressions are
\[
  \SNR_Y \to \tau_{\delta}^{-1}p\frac{\mu^{(1)}}{\Sigma_{11}}, \quad \text{and} \quad \SNR_D \to \tau_{\gamma}^{-1}p\frac{\mu^{(1)}}{\Sigma_{11}}.
  \]
And the $R^2$ of the $Y$ and $D$ regressions are
\[
R^2_Y \to \frac{\tau_{\delta}^{-1}p}{\tau_{\delta}^{-1}p + \Sigma_{11}/\mu^{(1)}} \quad \text{and} \quad R^2_D \to \frac{\tau_{\gamma}^{-1}p}{\tau_{\gamma}^{-1}p + \Sigma_{22}/\mu^{(1)}}.
\]
\end{pro}

\begin{proof}[Proof of \autoref{pro:rsquared}]
  As the two regressions are similar, let us focus on the $Y$ regression.
  Note that $\E[X_i]=0$. The \emph{prior-induced} variance of the signal (per observation) is
\[
\var(X_i'\delta) = \E\left[X_i'\var(\delta)X_i\right] = \tau_{\delta}^{-1} \tr(\Sigma_X)\rightarrow\tau_{\delta}^{-1}p\mu^{(1)}.
\]
The per-observation noise variance is simply the error variance for the $Y$ equation, which is $\Sigma_{11}$. The signal-noise ratio is then
\[
\SNR_Y \to \tau_{\delta}^{-1}p\frac{\mu^{(1)}}{\Sigma_{11}}.
\]

Then $R^2$ is the fraction of the total variance that is explained by the signal. For the $Y$ equation, one can write
\[
R^2 = \frac{\text{Signal Variance}}{\text{Signal Variance} + \text{Noise Variance}} 
=\frac{\SNR_Y}{\SNR_Y + 1} = \frac{\tau_{\delta}^{-1}p}{\tau_{\delta}^{-1}p + \Sigma_{11}/\mu^{(1)}}.
\]
\end{proof}

\begin{proof}[Proof of \autoref{lem:alpha-prior}]
From the properties of the Inverse-Wishart distribution, 
\begin{align*}
\alpha|\Sigma_{11\cdot2} &\sim \normal\left(\frac{\Sigma_{0,12}}{\Sigma_{0,22}}, \frac{\Sigma_{11\cdot2}}{\Sigma_{0,22}}\right)\\
\Sigma_{11\cdot2} &\sim \IG\left(\frac{\nu_0}{2}, \frac{\Sigma_{0,11\cdot2}}{2}\right).
\end{align*}
where \(\Sigma_{11\cdot2}= \Sigma_{11} - \Sigma_{12}^2/\Sigma_{22}\) and \(\Sigma_{0,11\cdot2} = \Sigma_{0,11} - \Sigma_{0,12}^2/\Sigma_{0,22}\).
Thus, the marginal distribution of \(\alpha\) is a location-scale \(t\)-distribution:
\begin{align*} 
  \alpha &\sim \frac{\Sigma_{0,12}}{\Sigma_{0,22}} + \frac{t_{\nu_0}}{\sqrt{\nu_0\Sigma_{0,22}/\Sigma_{0,11\cdot2}}}\\ 
  &= \frac{\Sigma_{0,12}}{\Sigma_{0,22}} + \frac{|\Sigma_0|^{1/2}}{\sqrt{\nu_0}\Sigma_{0,22}}t_{\nu_0}.
\end{align*}
\end{proof}

The following two propositions generalize \autoref{pro:sel-bias} from \autoref{sec:prior} to the case where $D_i$ may not be binary.
In this setting we work with the \emph{confounding bias}, defined as \(\Delta(z)=\E[Y_i|D_i=z]-\E[Y_i(z)]\), rather than selection bias but the intuition and proofs are effectively identical.

\begin{pro}[Confounding bias: na\"ive estimator] \label{pro:conf-bias-naive}
  Assume that the model and na\"ive estimator are as above. Suppose that Assumptions \ref{assump:dgp}--\ref{assump:rates} hold. The confounding bias of the na\"ive estimator is
  \[
  \sqrt{p}\Delta(z) \dto \normal\left(0,\frac{\lambda^{-1} (d^{(1)})^2}{\tau_{\gamma}^{-1}d^{(2)}}z^2\right).
  \]
  Therefore, \(\Delta(z)\pto0\), if \(p\to\infty\) as \(n\to\infty\), and thus the na\"ive estimator is dogmatic in the limit.
\end{pro}

\begin{proof}[Proof of \ref{pro:conf-bias-naive}]
  See \cite{linero_nonparametric_2023} Supplement S.2. Note that
  \[
  \Delta(z) = z\frac{\gamma'\Sigma_X\beta}{\Sigma_{22}+\gamma'\Sigma_X\gamma}.
  \]
\end{proof}

\begin{pro}[Confounding bias: BDML] \label{pro:conf-bias-BDML}
  Assume that the model and prior are as above. Suppose that Assumptions \ref{assump:dgp}--\ref{assump:rates} hold. 
  Then confounding bias of the BDML estimator is
  \[
  \Delta(z) \pto \frac{\Sigma_{12}}{\gamma'\Sigma_X\gamma+\Sigma_{22}}z,
  \]
  which is non-dogmatic in the limit.
\end{pro}

\begin{proof}[Proof of \autoref{pro:conf-bias-BDML}]
Note that the confounding bias is defined as \(\Delta(z)=\E[Y_i|D_i=z]-\E[Y_i(z)]\). For the first term, as \(\E[X_i]=0\), 
\begin{align*}
    \E[Y_i|D_i=z] &=\frac{\cov(Y_i,D_i)}{\var(D_i)} z=\frac{\cov(X_i\delta,X_i\gamma)+\cov(V_i,U_i)}{\var(X_i\gamma)+\var(V_i)} z\\ 
    &= \frac{\gamma'\Sigma_X\delta+\Sigma_{12}}{\gamma'\Sigma_X\gamma+\Sigma_{22}} z \pto \frac{\Sigma_{12}}{\gamma'\Sigma_X\gamma+\Sigma_{22}}z,
\end{align*} 
where the convergence follows from a similar argument as in Proposition \ref{pro:conf-bias-naive}. The distribution of the first term is non-degenerate as \(\Sigma\sim \IW(\nu_0,\Sigma_0)\).
The second term equals 0 as \(D\) does not directly enter into the \(Y\) equation. Combining both terms, we obtain the distribution of the confounding bias under the BDML estimator, which is non-dogmatic in the limit.
\end{proof}

\begin{proof}[Proof of \autoref{pro:post-mean-naive}]
Denote the posterior mean estimator of \(\alpha\) as \(\hat\alpha_\naive\). Denote \(p/n=r(1+o(1))\). Let 
\[
S_X = X'X/n \quad \text{and} \quad S_{\lambda} = S_X + \lambda/n\I.
\]
Recall \(\Sigma_X=Q_X\Lambda_XQ_X'\) and \(Q_X\) is the matrix of eigenvectors of \(\Sigma_X\). For a generic vector \(\theta\), let \(\theta_Q=Q_X'\theta\).
By an argument similar to that of \autoref{pro:ridgeMSE},
  \begin{align*}
    \hat\alpha_\naive&=\frac{D'(\I_n-XS_{\lambda}^{-1}X'/n)Y}{D'(\I_n-XS_{\lambda}^{-1}X'/n)D}
    = \frac{(X\gamma^*+V)'(\I_n-XS_{\lambda}^{-1}X'/n)(\alpha^*D+ X\beta^* +\varepsilon)}{(X\gamma^*+V)'(\I_n-XS_{\lambda}^{-1}X'/n)(X\gamma^*+V)}\\ 
    &= \alpha^*+\frac{(X\gamma^*+V)'(\I_n-XS_{\lambda}^{-1}X'/n)(X\beta^* +\varepsilon)}{(X\gamma^*+V)'(\I_n-XS_{\lambda}^{-1}X'/n)(X\gamma^*+V)}.
  \end{align*}  

  For the numerator, note that \(\|\delta^*\|_2\) and \(\|\gamma^*\|_2\) are bounded, and since \(\beta^*=\delta^*-\alpha^*\gamma^*\), it follows that \(\|\beta^*\|_2\) is bounded as well. Moreover, the spectral of \(\Sigma_X\) converge to a limit distribution on \((0,\infty)\). Therefore, following the argument in Lemma S.2 of \citet{linero_nonparametric_2023}, we have that 
  \[
  V'(\I_n-XS_{\lambda}^{-1}X'/n)X\beta^*/n,\; \gamma^{*\prime}X'(\I_n-XS_{\lambda}^{-1}X'/n)\varepsilon/n,\; V'(\I_n-XS_{\lambda}^{-1}X'/n)\varepsilon/n
  \]
  are all \(O_p\left(\frac 1 {\sqrt n}\right)\). There is only one remaining term
  \begin{align*}
    \gamma^{*\prime}X'(\I_n-XS_{\lambda}^{-1}X'/n)X\beta^*/n &= \frac{\lambda} n\sum_{j=1}^{p}\frac{d_j}{d_j+\lambda/n} \gamma^*_{Q,j}\beta^*_{Q,j}+O_p\left(\frac 1 {\sqrt n}\right).
  \end{align*}
 For the denominator, similarly, we have \(V'(\I_n-XS_{\lambda}^{-1}X'/n)X\gamma^*/n\) is \(O_p\left(\frac 1 {\sqrt n}\right)\), and two remaining terms
  \begin{align*}
    \gamma^{*\prime}X'(\I_n-XS_{\lambda}^{-1}X'/n)X\gamma^*/n &= \frac{\lambda} n\sum_{j=1}^{p}\frac{d_j}{d_j+\lambda/n} (\gamma_{Q,j}^*)^2+O_p\left(\frac 1 {\sqrt n}\right),\\
    V'(\I_n-XS_{\lambda}^{-1}X'/n)V/n &= \Sigma^*_{22}\left(\frac{\lambda} n\sum_{j=1}^{p}\frac{1}{d_j+\lambda/n}+1-r\right)+O_p\left(\frac 1 {\sqrt n}\right).
  \end{align*}
  All \(O_p(1/\sqrt n)\) terms are mean zero.

  Combining the numerator and denominator, we have
  \begin{align}
    \hat\alpha_\naive-\alpha^* = \frac{\frac{\lambda} n\sum_{j=1}^{p}\frac{d_j}{d_j+\lambda/n} \gamma^*_{Q,j}\beta^*_{Q,j}+O_p\left(\frac 1 {\sqrt n}\right)}{\frac{\lambda} n\sum_{j=1}^{p}\frac{d_j}{d_j+\lambda/n} (\gamma_{Q,j}^*)^2+\Sigma^*_{22}\left(\frac{\lambda} n\sum_{j=1}^{p}\frac{1}{d_j+\lambda/n}+1-r\right)+O_p\left(\frac 1 {\sqrt n}\right)}\label{eq:naive-fc}
  \end{align}
  As \(\lambda=o(n)\), and \(\|\beta^*\|_2\) and \(\|\gamma^*\|_2\) are bounded,
  \begin{align*}
    \hat\alpha_\naive-\alpha^* &= \frac 1 {\Sigma^*_{22}}\frac{\lambda} n\sum_{j=1}^{p}\frac{d_j}{d_j+\lambda/n} \gamma^*_{Q,j}\beta^*_{Q,j}+O_p\left(\frac 1 {\sqrt n}\right)\\
    &= \frac{\lambda} n\frac {\gamma^{*\prime}\beta^*} {\Sigma^*_{22}}+o\left(\frac {p} {n}\right)+O_p\left(\frac 1 {\sqrt n}\right).
  \end{align*} 
  where the \(O_p(1/\sqrt n)\) term is mean zero, so the na\"ive estimator is consistent. The last line is given by the fact that \(\lambda \asymp p\).
  \begin{align*}
    \sqrt n(\hat\alpha_\naive-\alpha^*)=\frac{\lambda} {\sqrt{n}}\frac {\gamma^{*\prime}\beta^*} {\Sigma^*_{22}}+o\left(\frac {p} {\sqrt{n}}\right)+O_p\left(1\right).
  \end{align*}
As \(\lambda\asymp p\), if \(\lim_{n\to\infty}p^2/n=\infty\) and 
  \(\gamma^{*\prime}\beta^*\neq0\),  the na\"ive estimator is not \(\sqrt n\)-consistent.
\end{proof}

\begin{proof}[Proof of \autoref{pro:marginal-post-Sigma}]
Now, let's derive the (approximate) marginal posterior distribution of \(\Sigma\) integrating out \(B\). 
We begin with the joint posterior of \(B\) and \(\Sigma\):
\begin{align*}
p(B, \Sigma | X, W) &\propto |\Sigma|^{-(n + \nu_0 + 3)/2} \exp\left(-\frac{1}{2}\tr(\Sigma_0\Sigma^{-1})\right) \\
&\quad \times \exp\left(-\frac{1}{2}\tr[(W - XB)'(W - XB)\Sigma^{-1}]\right) \\
&\quad \times \exp\left(-\frac{1}{2}\tvec(B)' V_0^{-1}\tvec(B) \right).
\end{align*}
Integrating out \(B\), we get:
\begin{align*}
p(\Sigma | X, W) &\propto |\Sigma|^{-(n + \nu_0 + 3)/2} |V_n|^{1/2} \exp\left(-\frac{1}{2}\tr\left((\Sigma_0+WW')\Sigma^{-1}\right)\right) \\
&\quad \times \exp\left(\frac{1}{2}\tvec(B_n)'V_n^{-1}\tvec(B_n)\right). 
\end{align*}
Based on Assumption \ref{assump:rates}, \(\tau_{\delta},\tau_{\gamma} = o(n)\). Then, for large \(n\), we can approximate \(V_n\):
\begin{align}
  V_n&= (\Sigma^{-1} \otimes X'X + V_0^{-1})^{-1} =\Sigma \otimes (X'X)^{-1}(1+o_p(1)), \label{eq:vn-0-order}\\
V_n &= \Sigma \otimes (X'X)^{-1} - (\Sigma^{-1} \otimes X'X)^{-1}V_0^{-1}(\Sigma^{-1} \otimes X'X)^{-1}(1+o_p(1)) \label{eq:vn-1st-order} 
\end{align}
where \eqref{eq:vn-0-order} and \eqref{eq:vn-1st-order} give the zeroth and first order approximations, respectively.

First, we approximate \(|V_n|^{1/2}\) using the first order approximation \eqref{eq:vn-1st-order}:
\begin{align*}
|V_n|^{1/2} &= |\Sigma \otimes (X'X)^{-1}|^{1/2}\left|I - V_0^{-1}(\Sigma \otimes (X'X)^{-1})\right|^{1/2}(1+o_p(1)) \nonumber \\
&= |\Sigma|^{p/2} |(X'X)^{-1}| \exp\left(-\frac{1}{2}\tr(V_0^{-1}(\Sigma \otimes (X'X)^{-1}))\right)(1+o_p(1))\\
&= |\Sigma|^{p/2} |(X'X)^{-1}| \exp\left(-\frac{1}{2} (\tau_{\delta}\Sigma_{11} + \tau_{\gamma}\Sigma_{22}) \tr((X'X)^{-1})\right)(1+o_p(1)),
\end{align*}
where the second line is given by the fact that \(|I-A|\approx1-\tr(A)\approx\exp(-\tr(A))\) for small \(A\). Note that 
\begin{align}
  (\tau_{\delta}\Sigma_{11} + \tau_{\gamma}\Sigma_{22}) \tr((X'X)^{-1})
  &= \tr\left(\begin{bmatrix}
    \tau_{\delta}\tr((X'X)^{-1}) & 0 \\ 0 & \tau_{\gamma}\tr((X'X)^{-1})
  \end{bmatrix}\Sigma\right). \label{eq:tr-sigma1}
\end{align}

Second, note that
\begin{align*}
\tvec(B_n) &= \left(I-V_nV_0^{-1}\right)\tvec(\hat{B}). 
\end{align*}
Now, let's expand \(\tvec(B_n)'V_n^{-1}\tvec(B_n)\):
\begin{align*}
&\tvec(B_n)'V_n^{-1}\tvec(B_n) \\
&=\tvec(\hat{B})'\left(I-V_nV_0^{-1}\right)' V_n^{-1} \left(I-V_nV_0^{-1}\right)\tvec(\hat{B}) \\
&=\tvec(\hat{B})'V_n^{-1}\tvec(\hat{B})- 2\tvec(\hat{B})'V_0^{-1}\tvec(\hat{B})  +\tvec(\hat{B})'V_0^{-1}V_nV_0^{-1}\tvec(\hat{B})\\
&=\tvec(\hat{B})'( \Sigma^{-1} \otimes X'X)\tvec(\hat{B})- \tvec(\hat{B})'V_0^{-1}\tvec(\hat{B}) \\ 
&\quad +\tvec(\hat{B})'V_0^{-1}(\Sigma \otimes (X'X)^{-1})V_0^{-1}\tvec(\hat{B})(1+o_p(1)),
\end{align*}
where the last line is by plugging in \(V_n^{-1}=\Sigma^{-1} \otimes X'X+V_0^{-1}\) and
the zeroth order approximation of \(V_n\) in \eqref{eq:vn-0-order}. The first term \(\tvec(\hat{B})'(\Sigma^{-1} \otimes X'X)\tvec(\hat{B})\) will be absorbed by the OLS part in \(\tilde{\Sigma}_n\). The second term \(\tvec(\hat{B})'V_0^{-1}\tvec(\hat{B})\) does not involve \(\Sigma\). Let \(\hat B_V=[\tau_{\delta}\hat\delta,\tau_{\gamma}\hat\gamma]\). Then, \(\tvec(\hat{B}_V)=V_0^{-1}\tvec(\hat{B})\), and the third term 
\begin{align}
\tvec(\hat{B})'V_0^{-1}(\Sigma \otimes (X'X)^{-1})V_0^{-1}\tvec(\hat{B}) &=\tvec(\hat{B}_V)'(\Sigma \otimes (X'X)^{-1})\tvec(\hat{B}_V)\label{eq:tr-sigma2}\\
&=\tr\left(\hat B_V'(X'X)^{-1}\hat B_V\Sigma\right),\notag
\end{align}
where 
\begin{align}
\hat B_V'(X'X)^{-1}\hat B_V&=\begin{bmatrix}
  \tau_{\delta}\hat\delta' \\ \tau_{\gamma}\hat\gamma'
\end{bmatrix}(X'X)^{-1}\begin{bmatrix}
  \tau_{\delta}\hat\delta & \tau_{\gamma}\hat\gamma
\end{bmatrix}\label{eq:tilde-C2}\\
&=\begin{bmatrix}
  \tau_{\delta}^2\hat\delta'(X'X)^{-1}\hat\delta & \tau_{\delta}\tau_{\gamma}\hat\delta'(X'X)^{-1}\hat\gamma \\
  \tau_{\delta}\tau_{\gamma}\hat\gamma'(X'X)^{-1}\hat\delta & \tau_{\gamma}^2\hat\gamma'(X'X)^{-1}\hat\gamma
\end{bmatrix}.\notag
\end{align} 
Combining \eqref{eq:tr-sigma1} and \eqref{eq:tr-sigma2}, we obtain the tilting part in \(\tilde C_n\).

\end{proof}

\begin{proof}[Proof of \autoref{pro:post-BDML}]
  First,  we approximate the posterior mode (and mean) of \(\Sigma\) by the fixed-point iteration solution to the Riccati equation. From the FOC, we have the matrix quadratic equation
    \[
    \Sigma\tilde C_n\Sigma - (\tilde\nu_n + 3)\Sigma + \tilde{\Sigma}_n = 0.
    \]
    Rearranging into fixed-point form
    \[
    \Sigma = \frac{\tilde{\Sigma}_n}{\tilde\nu_n + 3} + \frac{\Sigma\tilde C_n\Sigma}{\tilde\nu_n + 3}.
    \]
    Taking the zeroth-order approximation \(\Sigma^{(0)} = \tilde{\Sigma}_n/(\tilde\nu_n + 3)\) and performing one Picard iteration, we have
    \begin{align*}
    \Sigma^{(1)} &= \frac{\tilde{\Sigma}_n}{\tilde\nu_n + 3} + \frac{1}{\tilde\nu_n + 3}\left(\frac{\tilde{\Sigma}_n}{\tilde\nu_n + 3}\right)\tilde C_n\left(\frac{\tilde{\Sigma}_n}{\tilde\nu_n + 3}\right)\\
    &= \frac{\tilde{\Sigma}_n}{n} + \frac{\tilde{\Sigma}_n\tilde C_n\tilde{\Sigma}_n}{n^3}(1+o_p(1)),
    \end{align*}
where the last line is by the fact that \(\tilde\nu_n + 3 \sim n\).

  Note that
  \begin{align*}
    \frac{\tilde{\Sigma}_n}{n} &= \frac{\Sigma_0}{n}+\frac{(W - X\hat{B})'(W - X\hat{B})}{n}\\ 
    &= \frac{\Sigma_0}{n}+\Sigma^*-\frac p n \Sigma^*+O_p\left(\frac 1 {\sqrt n}\right),
    \end{align*}
    where the \(O_p\left(\frac 1 {\sqrt n}\right)\) term is mean zero.
    The first line is by the definition of \(\tilde{\Sigma}_n\). The second line is by the property of sample variance covariance matrix of the OLS residues, i.e., \(\E\left[{(W - X\hat{B})'(W - X\hat{B})}/{n-p}\right]= \Sigma^*,\) given \(p<n\) and \(p=o(n)\) in Assumptions \ref{assump:controls} and \ref{assump:rates}.
Under Assumptions \ref{assump:dgp}--\ref{assump:rates}, we can also approximate \(\tilde C_n\) as
  \begin{align*}
    \tilde C_n &= C^*(1+o_p(1)),\text{ where }C^*\asymp p^2/n.
    \end{align*}
    Then, the posterior mode and mean of \(\Sigma\) up to the order of \(1/n\) and \((p/n)^2\) is\footnote{Here we consider \((p/n)^2\) as the \(p/n\) order terms will be cancelled out in the approximation for \(\hat\alpha_\BDML\).}
\begin{align*}
  \E[\Sigma|X,W] &=\Sigma^*+\frac{1}{n}\Sigma_0-\frac p n \Sigma^*+\frac 1 n\Sigma^*C^*\Sigma^*+o_p\left(\frac 1 n\right)+o_p\left(\left(\frac p n\right)^2\right)+O_p\left(\frac 1 {\sqrt n}\right),
\end{align*}
where the \(O_p\left(\frac 1 {\sqrt n}\right)\) term is mean zero.

Also note that the posterior variance of \(\Sigma\) is of order \(O_p\left(\frac 1 n\right)\), and \(\Sigma^*_{22}>0\). Then, the posterior mean of \(\alpha\) is approximately
\begin{align*}
  \hat\alpha_\BDML=\E[\alpha|X,W] & =\frac{\Sigma^*_{12}+[\Sigma_{0,12}-p\Sigma^*_{12}+(\Sigma^*C^*\Sigma^*)_{12}]/n}{\Sigma^*_{22}+[\Sigma_{0,22}-p\Sigma^*_{22}+(\Sigma^*C^*\Sigma^*)_{22}]/n}\\ &\quad+o_p\left(\frac 1 n\right)+o_p\left(\left(\frac p n\right)^2\right)+O_p\left(\frac 1 {\sqrt n}\right),
     \end{align*}
     where the \(O_p\left(\frac 1 {\sqrt n}\right)\) term is mean zero. Then, we approximate \(\hat\alpha_\BDML\) up to the order of \(1/n\) and \((p/n)^2\). Let
     \begin{align*}
      \Delta &= [\Sigma_0- p \Sigma^*+\Sigma^*C^*\Sigma^*]/n,
     \end{align*}
     Then,
     \begin{align*}
      &\frac{\Sigma^*_{12}+[\Sigma_{0,12}-p\Sigma^*_{12}+(\Sigma^*C^*\Sigma^*)_{12}]/n}{\Sigma^*_{22}+[\Sigma_{0,22}-p\Sigma^*_{22}+(\Sigma^*C^*\Sigma^*)_{22}]/n} = \frac{\Sigma^*_{12}+\Delta_{12}}{\Sigma^*_{22}+\Delta_{22}} \\&= \frac{\Sigma^*_{12}}{\Sigma^*_{22}}\left[1+\left(\frac{\Delta_{12}}{\Sigma^*_{12}}-\frac{\Delta_{22}}{\Sigma^*_{22}}\right)\left(1-\frac{\Delta_{22}}{\Sigma^*_{22}}\right)\right]
      \end{align*}
      Note that \begin{align*}\frac{\Delta_{12}}{\Sigma^*_{12}}-\frac{\Delta_{22}}{\Sigma^*_{22}} &= \frac{\Sigma_{0,12}-p\Sigma^*_{12}+(\Sigma^*C^*\Sigma^*)_{12}}{n\Sigma^*_{12}}-\frac{\Sigma_{0,22}-p\Sigma^*_{22}+(\Sigma^*C^*\Sigma^*)_{22}}{n\Sigma^*_{12}}\\ 
      &= \frac 1 n\left(\frac{\Sigma_{0,12}+(\Sigma^*C^*\Sigma^*)_{12}}{\Sigma^*_{12}}-\frac{\Sigma_{0,22}+(\Sigma^*C^*\Sigma^*)_{22}}{\Sigma^*_{12}}\right)\\ &=O\left(\frac 1 n\right)+O\left(\left(\frac p n\right)^2\right),\end{align*} where we cancel out the \(p/n\) order terms.
     As \(\alpha^*=\Sigma^*_{12}/\Sigma^*_{22}\), we have
      \begin{align*}
        \hat\alpha_\BDML &= \alpha^*+\frac {\alpha^*} n\left(\frac{\Sigma_{0,12}+(\Sigma^*C^*\Sigma^*)_{12}}{\Sigma^*_{12}}-\frac{\Sigma_{0,22}+(\Sigma^*C^*\Sigma^*)_{22}}{\Sigma^*_{12}}\right)\\ &\quad+o_p\left(\frac 1 n\right)+o_p\left(\left(\frac p n\right)^2\right)+O_p\left(\frac 1 {\sqrt n}\right),
         \end{align*}
         where we cancel out common terms in the second equality.
    Then, \(\hat\alpha_\BDML\pto\alpha^*\), so the BDML estimator is consistent. 
    
    Moreover,
\begin{align*}
   \sqrt n(\hat\alpha_\BDML-\alpha^*) &=\frac {\alpha^*} {\sqrt{n}}\left(\frac{\Sigma_{0,12}+(\Sigma^*C^*\Sigma^*)_{12}}{\Sigma^*_{12}}-\frac{\Sigma_{0,22}+(\Sigma^*C^*\Sigma^*)_{22}}{\Sigma^*_{22}}\right)\\ &\quad+o_p\left(\frac 1 {\sqrt n}\right)+o_p\left(\frac {p^2}{n^{3/2}}\right)+O_p\left(1\right),
\end{align*}
where the \(O_p(1)\) term is mean zero. Then, if \(p=o(n^{3/4})\), \begin{align*}
  \frac {\alpha^*} {\sqrt{n}}\left(\frac{\Sigma_{0,12}+(\Sigma^*C^*\Sigma^*)_{12}}{\Sigma^*_{12}}-\frac{\Sigma_{0,22}+(\Sigma^*C^*\Sigma^*)_{22}}{\Sigma^*_{22}}\right)=O\left(\frac 1 {\sqrt n}\right)+O\left(\frac {p^2}{n^{3/2}}\right)=o_p(1),
\end{align*}
then \(\hat\alpha_\BDML\) is \(\sqrt n\)-consistent.
\end{proof}

\begin{proof}[Proof of \autoref{pro:var-alpha}]
  Both \(\hat\alpha_\BDML\) and \(\hat\alpha_\FDML\) can be approximated as \(\hat\alpha \approx \bar\Sigma_{12}/\bar\Sigma_{22}\) where \(\bar\Sigma = (W - X\hat{B})'(W - X\hat{B})/n\). Applying the delta method to this ratio, we obtain the leading term of their asymptotic variance
  \begin{align*}
  &n\var\left[\frac{\Sigma^*_{12}\bar\Sigma_{22}-\Sigma^*_{22}\bar\Sigma_{12}}{(\Sigma^*_{22})^2}\right]\\
  &=\frac n {(\Sigma^*_{22})^4}\left[(\Sigma^*_{12})^2\var(\bar\Sigma_{22})+(\Sigma^*_{22})^2\var(\bar\Sigma_{12})-2\Sigma^*_{12}\Sigma^*_{22}\cov(\bar\Sigma_{12},\bar\Sigma_{22})\right](1+o(1)).
  \end{align*}
Note that \(n\bar\Sigma\sim \mathcal{W}_2(n-p, \Sigma^*)\), then we have
  \[
    \cov\left(n\bar\Sigma_{ij},n\bar\Sigma_{kl}\right)=(n-p)\bigl(\Sigma^*_{ik}\Sigma^*_{jl}+\Sigma^*_{il}\Sigma^*_{jk}\bigr).
  \]
 As $p=o(n)$, we have
  \begin{align*}
    n\var(\bar\Sigma_{22}) &\to 2(\Sigma^*_{22})^{2},\\
    n\var(\bar\Sigma_{12}) &\to \Sigma^*_{11}\Sigma^*_{22}+(\Sigma^*_{12})^{2},\\
    n\cov(\bar\Sigma_{12},\bar\Sigma_{22}) &\to 2\Sigma^*_{12}\Sigma^*_{22}.
  \end{align*}
  Results follow from plugging in the above expressions.

  For \(\hat\alpha_\naive\), the leading term of the asymptotic variance is given by
  \begin{align*}
    n\var\left[\frac{V'(\I_n-XS_{\lambda}^{-1}X'/n)\varepsilon/n}{\Sigma^*_{22}}\right]=\frac {(\sigma^*_{\varepsilon})^2\Sigma^*_{22}}{(\Sigma^*_{22})^2}(1+o(1))=\frac {(\sigma^*_{\varepsilon})^2}{\Sigma^*_{22}}(1+o(1)),
  \end{align*}
  where the first equality is given by the fact that \begin{align*}
    n\var[V'(\I_n-XS_{\lambda}^{-1}X'/n)\varepsilon/n]&\to \var\left(\frac 1 n\sum_{i=1}^nV_i'\varepsilon_i\right)=(\sigma^*_{\varepsilon})^2\Sigma^*_{22}.
  \end{align*}
  Our structural model imposes \(\alpha^*=\Sigma^*_{12}/\Sigma^*_{22}\) and
\((\sigma^*_{\varepsilon})^2=\Sigma^*_{11}-(\alpha^{*})^2\Sigma^*_{22}=\Sigma^*_{11}-(\Sigma^*_{12})^{2}/\Sigma^*_{22}\), so the asymptotic variance expressions coincide. 
\end{proof}

\begin{proof}[Proof of \autoref{lem:matrix-bernstein}]
  See Theorems 2.8.1 and 5.4.1 in Vershynin (2018).
\end{proof}
\begin{proof}[Proof of \autoref{pro:bvm-known-sigma}]
  The proof follows from Walker (2024) by verifying his assumptions in our high dimensional regression setup.

  \medskip
\noindent\textbf{Part 1: Assumptions 1 and 2 in Walker (2024).} Assumption 1 in Walker (2024) is satisfied by our model setup and Assumptions \ref{assump:dgp}--\ref{assump:rates}. Our Assumption \ref{assump:controls} also implies that $\E[X_iX_i']$ is non-singular, which ensures identification.

  Assumption 2 in Walker (2024) is satisfied by our prior together with the additional assumption that $\sigma_{\varepsilon}^{2*}$ and $\sigma_{V}^{2*}$ are known, and the prior for $\alpha$ is continuous and positive over a neighborhood of $\alpha^*$.

  \medskip
  \noindent\textbf{Part 2: Assumptions 3 and 4 in Walker (2024).} For notation simplicity, let $\theta_1 = \delta$, $\theta_2 = \gamma$, and $\theta = (\theta_1',\theta_2')'$.
  By Corollary \ref{cor:eigen-bounds}, since $X_i$ is sub-gaussian and i.i.d., and the eigenvalues of $\Sigma_X$ is uniformly bounded by $[\underline\lambda_X,\bar\lambda_X]$ with $0<\underline\lambda_X\le \bar\lambda_X<\infty$. For event \[\mathcal K_n =\left\{X\in\mathbb R^{n\times p}:\,\lambda_{\min}\left(\frac{X'X} n\right)\ge \frac{\underline\lambda_X}2\text{ and }\lambda_{\max}\left(\frac{X'X} n\right)\le 2\bar\lambda_X\right\},\] we have that $P(\mathcal K_n)\pto 1$ as $n\to\infty$.

  Define $\mathcal M_n = \mathcal M$, then Walker (2024)'s Assumption 3.1 is automatically satisfied. 

  For Assumption 3.2 in Walker (2024), the posterior contracts at rate
  \[
  \|\theta-\theta^*\|_2 = O_p\left(\sqrt{\frac{p}{n}}\right)
  \] uniformly over $X\in\mathcal K_n$.
  By the equivalence of norms, in event $\mathcal K_n$, we have
  \begin{align}
    \frac{\underline\lambda_X}2\|\theta-\theta^*\|_2^2 \le \|m-m^*\|_{n,2}^2 \le 2\bar\lambda_X\|\theta-\theta^*\|_2^2,\label{eq:equi-norm}
  \end{align}
  posterior contraction in \(\ell_2\)-norm implies contraction in the \(L_2(P_n)\) norm for the functions. Therefore, let $\nu_n=\sqrt{p/n}$, for some constants $D_m,D_{\theta}>0$,
  \begin{align*}
  &\E_{P^*}\left[P\left(\mathcal M_n \cap B_{\mathcal M,n,2}(m^*,D_m \nu_n)^c \mid X,W\right) \right] \\ 
  &\le \E_{P^*}\left[P\left(\|\theta-\theta^*\|_2>D_{\theta}\nu_n\mid X,W,\mathcal K_n\right)\right]+P(\mathcal K_n^c)\to 0.
  \end{align*}

  For Assumption 4, as $p=o\left(\sqrt{n}\right)$, $\nu_n=\sqrt{p/n}=o(n^{-1/4})$. 
  
  \medskip
  \noindent\textbf{Part 3: Assumption 5 in Walker (2024).} 
  Assumption 5.1 considers the empirical process
  \[
  G_n^{(1)}(m) = \sqrti V_i \left[ m_1(X_i)- m_1^*(X_i)\right]
  \]
  be uniformly small over the set $B_{\mathcal M,n,2}(m^*,D_m \nu_n)$.

  In our linear model, write
  \(
  m_1(X_i)- m_1^*(X_i)=X_i' (\delta-\delta^*).
  \)
  Thus, define the function class
  \[
  \mathcal{F}_{\nu} = \left\{ f(X) = X' (\delta-\delta^*) : \|\delta-\delta^*\|_2 \le \nu \right\}.
  \]
  The empirical \(L_2\) norm is defined by
  \[
  \|f\|_{n,2} = \left[ \avgi \left(X_i' (\delta-\delta^*)\right)^2\right]^{1/2}.
  \]
  
  Under event $\mathcal K_n$, by equivalence of norm \eqref{eq:equi-norm}, the metric on \(\mathcal{F}_{\nu}\) is equivalent to the Euclidean norm on \(\delta-\delta^*\).
  Then, the covering number of the ball \(B_{\delta,D_{\theta}\nu_n}=\{\delta: \|\delta-\delta^*\|_2 \le D_{\theta}\nu_n\}\) in \(\mathbb{R}^{p}\) satisfies
  \[
  N\left(\epsilon, B_{\delta,D_{\theta}\nu_n}, \|\cdot\|_2\right)
  \le \left(\frac{3D_{\theta}\nu_n}{\epsilon}\right)^{p},
  \]
  for $\epsilon<D_{\theta}\nu_n$.
  Thus, the covering number for \(\mathcal{F}_{D_m\nu_n}\) with respect to \(\|\cdot\|_{n,2}\) follows
  \[
  \log N\left(\epsilon, \mathcal{F}_{D_m\nu_n}, \|\cdot\|_{n,2}\right)
  \lesssim p \log\frac{3D_m\nu_n}{\epsilon}.
  \]
  Dudley’s inequality states that there is a constant \(K>0\) such that
  \begin{align*}
  \mathbb{E}\left[\sup_{f\in\mathcal{F}_{D_m\nu_n}} \left| \sqrti V_i f(X_i) \right|\cdot 1_{\mathcal K_n}\right]
  & \le K \int_0^{\infty} \sqrt{ \log N\left(\epsilon, \mathcal{F}_{D_m\nu_n}, \|\cdot\|_{n,2}\right)}\, d\epsilon\\ 
  &=K \int_0^{3D_m\nu_n} \sqrt{ \log N\left(\epsilon, \mathcal{F}_{D_m\nu_n}, \|\cdot\|_{n,2}\right)}\, d\epsilon,
  \end{align*}
  where the second line is by the fact that the covering number is one for \(\epsilon>3D_m\nu_n\).
  Then the entropy integral becomes
  \[
  \int_0^{3D_m\nu_n} \sqrt{p \log\frac{3D_m\nu_n}{\epsilon}}\, d\epsilon=O(\nu_n\sqrt p)=o(1).
  \]
  The first equality follows from that the integral \(\int_0^{1} \sqrt{\log(1/u)}\, du\) is finite, and the second equality is by $\nu_n = O(\sqrt{p/n})$ and $p=\left(\sqrt{n}\right)$.

  Under event $\mathcal K_n^c$, we do not have the norm equivalence. However, for any \(f\in\mathcal{F}_{D_m \nu_n}\), we have that
  \begin{align*}
  |G_n^{(1)}(m)| &\le \frac{1}{\sqrt{n}} \sum_{i=1}^n |V_i|\,\|X_i\|_2D_{\theta} \nu_n\\
  &\le D_{\theta} \nu_n\sqrt{n}\cdot\max_{1\le i\le n} |V_i|\cdot\max_{1\le i\le n}\|X_i\|_2,
  \end{align*}
  Under our sub-Gaussian assumptions for \(V_i\) and \(X_i\), Exercise 2.5.10 in Vershynin (2018) yield
  \[
    \max_{1\le i\le n} |V_i| = O_p(\sqrt{\log n})\text{ and } \max_{1\le i\le n} \|X_i\|_2 = O_p(\sqrt{\log n}).
  \]
  Hence, on \(\mathcal K_n^c\),
  \[
  \sup_{f\in\mathcal{F}_{D_m \nu_n}} |G_n^{(1)}(m)| \le D_{\theta} \nu_n\sqrt{n}\cdot O_p(\log n).
  \]
  Then,
  \begin{align*}
  \E\left[\sup_{f\in\mathcal{F}_{D_m \nu_n}} |G_n^{(1)}(m)|1_{\mathcal K_n^c}\right] &\le D_{\theta} \nu_n\sqrt{n}\cdot O(\log n)\cdot P(\mathcal K_n^c)\\ 
  &= D_{\theta} \nu_n\sqrt{n}\cdot O(\log n)\cdot O(p^{-1})\\ 
  & = O\left(\frac{\log n}{\sqrt{p}}\right),
  \end{align*}
  where the second line follows from \(P(\mathcal K_n^c) = O(p^{-1})\) by Corollary \ref{cor:eigen-bounds}, and the last line is by \(\nu_n = \sqrt{p/n}\).
  
  Thus, the overall bound on the supremum is given by
  \[
  \E\left[\sup_{f\in\mathcal{F}_{D_m \nu_n}} |G_n^{(1)}(m)|\right] \le O\left(\frac{p}{\sqrt{n}}\right) + O\left(\frac{\log n}{\sqrt{p}}\right).
  \]
  Hence, provided that
  \(p = o(\sqrt{n})\) and $\log n = o(p^{1/2})$, the overall supremum tends to zero. 

  This completes the verification of Walker’s Assumption 5.1 in the \(L_2(P_n)\) metric. Assumption 5.2 can be verified in a similar manner.

\end{proof}

\begin{proof}[Proof of \autoref{cor:eigen-bounds}]
  As \(X_i\) is sub-gaussian, 
\(
S_i = X_iX_i' - \Sigma_X
\)
is sub-exponential, and its Orlicz \(\psi_1\) norm is bounded, i.e., there exists a constant \(M_S>0\) such that
\(
\|S_i\|_{\psi_1} \le M_S,
\)
for all \(i\). By Lemma \ref{lem:matrix-bernstein}, for any $\epsilon \ge 0$,
\[
  P\left\{ \left\|\frac{X'X}n-\Sigma_X\right\| \ge \epsilon \right\}=P\left\{ \left\|\sumi S_i\right\| \ge tn \right\}
\le 2p \cdot \exp\left( - cn\cdot\min\left( \frac{\epsilon^2}{M_S^2},\,\frac{\epsilon}{M_S} \right)\right).
\]

Typically, $\epsilon$ is small, so the mininum is achieved by the quadratic term. We choose $\epsilon=M_S\sqrt{\frac{2\log p}{\tilde c n}}$, then
\[
  P\left\{ \left\|\frac{X'X}n-\Sigma_X\right\| \ge \epsilon \right\}
\le 2p \cdot \exp(-2\log p)=\frac 2 p=O(p^{-1})\to0,
\]
as $n\to\infty$. 

Define event $\mathcal E_n =\left\{X\in\mathbb R^{n\times p}:\,\left\|\frac{X'X}n-\Sigma_X\right\| < \epsilon\right\}$. For any $X\in\mathcal E_n$,
 Weyl's inequality implies that for every eigenvalue in decreasing order,
\(
\left|\lambda_j\left(\frac{X'X} n\right)- \lambda_{j}(\Sigma_X)\right|<\epsilon.\)
For large $n$, \(\epsilon<\underline\lambda_X/2<\bar\lambda_X\), so $\lambda_{\min}\left(\frac{X'X} n\right)\ge \underline\lambda_X/2$ and $\lambda_{\max}\left(\frac{X'X} n\right)\le 2\bar\lambda_X$, and thus $X\in\mathcal K_n$. Therefore, 
\begin{align*}
  P(\mathcal K_n)&\ge P(\mathcal E_n)\pto 1,\\ 
  P(\mathcal K_n^c)&\le P(\mathcal E_n^c)=O(p^{-1}),
\end{align*} as $n\to\infty$.
\end{proof}

%\pagebreak
%\normalsize
%\input{additional-figs.tex}

\end{document}